\documentclass{amsart}
\usepackage[margin=1in]{geometry}
\usepackage{hyperref}
\newcommand*{\mailto}[1]{\href{mailto:#1}{\nolinkurl{#1}}}
\newcommand{\arxiv}[1]{\href{http://arxiv.org/abs/#1}{arXiv:#1}}

\usepackage{curves}
\usepackage{enumerate}
\date{\today}
\newcommand{\R}{\mathbb{R}}

\newcommand{\N}{\mathbb{N}}
\newcommand{\C}{\mathbb{C}}
\newcommand{\Z}{\mathbb{Z}}

\newcommand{\A}{\mathcal{A}}

\newcommand{\be}{\begin{equation}}
\newcommand{\ee}{\end{equation}}
\newcommand{\bea}{\begin{eqnarray}}
\newcommand{\eea}{\end{eqnarray}}
\newcommand{\x}{{\rm x}}
\newcommand{\w}{{\rm w}}

\newcommand{\fr}{\frac}
\newcommand{\D}{\partial}
\newcommand{\al}{\alpha}
\newcommand{\ti}{\tilde}
\newcommand{\abs}[1]{\left\vert#1\right\vert}
\newcommand{\spr}[2]{\langle #1 , #2 \rangle}

\numberwithin{equation}{section}

\def\idty{{\mathchoice {\rm 1\mskip-4mu l} {\rm 1\mskip-4mu l} %
{\rm 1\mskip-4.5mu l} {\rm 1\mskip-5mu l}}}

\newtheorem{thm}{Theorem}[section]

\newtheorem{lem}[thm]{Lemma}

\newtheorem{Example}[thm]{Example}

\renewcommand{\epsilon}{\varepsilon}

\begin{document}

\title{Lieb--Robinson Bounds for the Toda Lattice}

\author[U. Islambekov]{Umar Islambekov}
\address{Department of Mathematics\\ University of Arizona\\
617 N. Santa Rita Avenue\\ Tucson, AZ 85721\\ USA}
\email{\mailto{uislambekov@math.arizona.edu}}
\urladdr{\url{http://math.arizona.edu/~uislambekov/}}

\author[R. Sims]{Robert Sims}
\address{Department of Mathematics\\ University of Arizona\\
617 N. Santa Rita Avenue\\ Tucson, AZ 85721\\ USA}
\email{\mailto{rsims@math.arizona.edu}}
\urladdr{\url{http://math.arizona.edu/~rsims/}}

\author[G.\ Teschl]{Gerald Teschl}
\address{Faculty of Mathematics\\ University of Vienna\\
Nordbergstrasse 15\\ 1090 Wien\\ Austria\\ and International Erwin Schr\"odinger
Institute for Mathematical Physics\\ Boltzmanngasse 9\\ 1090 Wien\\ Austria}
\email{\mailto{Gerald.Teschl@univie.ac.at}}
\urladdr{\url{http://www.mat.univie.ac.at/~gerald/}}

\thanks{J. Stat. Phys. {\bf 148}, 440--479 (2012)}
\thanks{R.\ S.\ and U.\ I.\ were supported, in part, by NSF grants DMS-0757424 and DMS-1101345.
G.\ T.\ was supported by the Austrian Science Fund (FWF) under Grant No.\ Y330.}


\begin{abstract}
We establish locality estimates, known as Lieb--Robinson bounds, for the Toda lattice. 
In contrast to harmonic models, the Lieb--Robinson velocity for these systems
do depend on the initial condition. Our results also apply to the entire Toda as 
well as the Kac-van Moerbeke hierarchy. Under suitable assumptions, our 
methods also yield a finite velocity for certain perturbations of these systems.  
\end{abstract}

\maketitle

%
%

\section{Introduction}

Analyzing the dynamics of certain non-relativistic systems is crucial to
understanding a number of important problems in statistical mechanics. 
For example, there is much interest in rigorously justifying the
emergence of macroscopic non-equilibrium phenomena, like Fourier's Law or
other forms of heat conduction, directly from a many-body Hamiltonian dynamics \cite{bonetto2000}.
Future progress on fundamental questions such as these will require detailed information on the
structure of the underlying dynamics. One of the goals of the present work is to investigate 
an approximate form of locality, often described by a Lieb--Robinson bound, for the 
dynamics corresponding to a class of integrable systems. 

Recognizing approximate forms of locality, despite the lack of a relativistic framework, has been 
essential in solving many intriguing open problems. One of the first mathematical formulations
of a useful locality estimate  was given by Lieb and Robinson in 1972 \cite{lieb1972}. In this work, they 
demonstrated that the dynamics corresponding to quantum spin systems, with e.g.\ finite-range interactions,
remains effectively confined to a "light" cone, up to corrections which decay at least exponentially 
away from the light cone.  Recently there have been a number of improvements and generalizations of
the original result \cite{nach12006, hast2006, nach22006, eisert2008, amour2009, nachtergaele2009, PS09, Pou10, schuch2010, nach2011}, 
and these new techniques have led to some interesting applications \cite{hastings2004a, hastings2007, NS2, nachtergaele2010, HS, bravyi:2010a, bravyi:2010b, bachmann2011, borovyk, hamza}.
For a review of these results, we refer the interested reader to \cite{NS3, hastings2010}.
 
Shortly after the original result of Lieb and Robinson, it was shown in \cite{marchioro1978}, see also \cite{butta2007, raz2009} for more recent
developments, that this notion of quasi-locality also applies to the dynamics of classical oscillator systems. 
Since it is this work that more closely pertains to the topic of the present article, we will discuss it briefly as follows.

Consider a system of particles confined to a large but finite set $\Lambda \subset \mathbb{Z}^d$. To each site
$x \in \Lambda$ associate a particle, or oscillator, with position $q_x \in \mathbb{R}$ and momentum $p_x \in \mathbb{R}$. 
The state of the system in $\Lambda$ is described by a sequence $\x = \{ (q_x, p_x) \}_{x \in \Lambda}$, and 
the set of all such sequences, $\mathcal{X}_{\Lambda}$, is called phase space.  A Hamiltonian, $H$, is a real-valued
function on phase space. Given a Hamiltonian and a sequence $\x \in \mathcal{X}_{\Lambda}$, Hamilton's equations of
motion are: for each $x \in \Lambda$,
\begin{equation} \label{eq:Hamseqns}
\dot{q}_x(t) = \frac{\partial H}{\partial p_x} \quad \mbox{and} \quad \dot{p}_x(t) = - \frac{\partial H}{\partial q_x}
\end{equation}
solved with initial condition $\{ (q_x(0), p_x(0) ) \}_{x \in \Lambda} = \x$. For many Hamiltonians, this system of
coupled differential equations can be solved for all time (and any initial condition). In this case, we denote by
$\Phi_t$, the Hamiltonian flow, i.e., the mapping that associates to initial conditions, the solution at time $t$:
$\Phi_t(\x) = \{ (q_x(t), p_x(t) ) \}_{x \in \Lambda}$. 

Measurements of the system under consideration correspond to observables, where an
observable $A$ is a complex-valued function on phase space. For example, the position 
of the particle at site $x \in \Lambda$ corresponds to an observable $Q_x$ for which
$Q_x(\x) = q_x$. Let us denote by $\mathcal{A}_{\Lambda}$ the set of all
observables over $\mathcal{X}_{\Lambda}$. Given a Hamiltonian for which 
\eqref{eq:Hamseqns} can be solved, the Hamiltonian dynamics $\alpha_t : \mathcal{A}_{\Lambda} \to \mathcal{A}_{\Lambda}$
is defined by $\alpha_t(A) = A \circ \Phi_t$. In this case, we can make time-dependent observations such as
$[\alpha_t(Q_x)](\x) = q_x(t)$.

It is well known that the Hamiltonian dynamics is generated by the Poisson bracket, i.e.,
\begin{equation}
\frac{d}{dt} \alpha_t(A) = \alpha_t \left( \left\{ A, H \right\} \right) = \left\{ \alpha_t(A), H \right\} \, ,
\end{equation}
and here the Poisson bracket of two observables is the observable
\begin{equation}
\left\{ A, B \right\} = \sum_{x \in \Lambda} \frac{\partial A}{\partial q_x} \cdot \frac{\partial B}{\partial p_x} - \frac{\partial A}{\partial p_x} \cdot \frac{\partial B}{\partial q_x}  \, .
\end{equation}
In terms of this Poisson bracket, we can now describe the Lieb--Robinson bound.

Let us fix finite sets $X,Y \subset \mathbb{Z}^d$ with $X \cap Y = \emptyset$. Take $\Lambda \subset \mathbb{Z}^d$
finite, but large enough so that $X \cup Y \subset \Lambda$.  Consider two observables $A, B \in \mathcal{A}_{\Lambda}$ 
with supports in $X, Y$ respectively, i.e., e.g.\ $A$ depends only on those $q_x$ and $p_x$ with $x \in X$. It is clear that
$\{ A, B \} =0$ since $A$ and $B$ have disjoint supports. A Hamiltonian $H$ satisfies a Lieb--Robinson bound if for
some initial condition $\x$ there exist numbers $\mu$, $C$, and $v$ for which
\begin{equation} \label{eq:genhamlrb}
\left| \left[ \left\{ \alpha_t(A), B \right\} \right] (\x) \right| \leq C e^{- \mu \left(d(X,Y) - v|t| \right)} \, ,
\end{equation} 
where $d(X,Y)$ is the distance between $X$ and $Y$. In words, this estimate shows that  for times $t$ with
$v|t| \leq d(X,Y)$ the support of $\alpha_t(A)$ remains essentially disjoint from the support of $B$, up to exponentially small
corrections in $d(X,Y)$. The number $v$ is often called the Lieb--Robinson velocity corresponding to $H$, and it represents a
bound on the rate at which disturbances can propagate through the system. 

We will now briefly discuss what is known concerning Lieb--Robinson bounds for classical systems. 
The works in \cite{marchioro1978, butta2007, raz2009} prove Lieb--Robinson type bounds for a variety of harmonic systems, and
it is shown that analogous results also hold for certain anharmonic perturbations. The strongest such result is in \cite{butta2007} and demonstrates that
for an anharmonic model with a quartic on-site perturbation the relevant Poisson bracket decays to zero whenever the distance between the
supports of the local observables grows faster than $t \log^{\alpha}(t)$ for suitable $\alpha >0$. In \cite{raz2009}, where explicit estimates on
the Lieb--Robinson velocity were obtained, it is shown that the number $v$ is independent of the initial condition $\x$. 
Since the system is linear, this is not so surprising.

Of course when it comes to anharmonic lattice oscillations one of the central objects is the famous
Fermi--Pasta--Ulam--Tsingou (FPU) problem. It was first demonstrated by Zabusky and Kruskal \cite{zakr} that one key ingredient
to resolve the FPU paradox is the relation with solitons. Moreover, subsequently Toda \cite{ta} presented an anharmonic
lattice which possess soliton solutions and was later on shown to be integrable by Flaschka \cite{Flaschka1,Flaschka2} and Manakov \cite{manakov1975}. Clearly this naturally raises the question about Lieb--Robinson type locality bounds for the Toda lattice and the main goal of this work is to establish such bounds. Our estimates
produce a Lieb--Robinson velocity that depends on the initial condition, see e.g.\ Theorem~\ref{thm:toda_lrb}. We discuss this
fact in the context of one-soliton solutions in Section~\ref{subsec:solitons}. For a restricted class of initial conditions, we can prove
the existence of a finite Lieb--Robinson velocity for a class of perturbations of the Toda lattice. We have two bounds of this
type. The first, in Section~\ref{subsec:direct}, establishes a result by directly mimicking the methods for the unperturbed system.  The next, in
Section~\ref{subsec:inter}, uses interpolation and shows that the velocity of the perturbed system can be estimated in terms of the velocity
of the unperturbed system. Analogous results are shown to also hold for solutions of the Toda hierarchy, see Section~\ref{sec:hth} and Section~\ref{sec:perthier}. For a much larger class
of initial conditions, we can prove a locality bound for perturbed systems, however, we do not establish the existence of a finite
Lieb--Robinson velocity. This is discussed in Section~\ref{sec:timedep}. The final section, Appendix~\ref{sec:app}, describes a set of perturbations and
corresponding initial conditions for which solutions remain globally bounded. This shows that the results obtained in
sections \ref{subsec:direct} and \ref{subsec:inter} are not vacuous.

%
%

\section{Locality Estimates for the Toda Lattice} \label{sec:loctoda}
In this section, we discuss the Toda Lattice and prove a Lieb--Robinson bound.
We introduce the model in Section~\ref{sec:toda}. The crucial solution estimate, 
Theorem~\ref{thm:toda_solest}, is contained in Section~\ref{subsec:solest} as well as some useful remarks. A 
Lieb--Robinson bound, see Theorem~\ref{thm:toda_lrb}, in terms of a large class of observables is
proven in Section~\ref{subsec:lrbtoda}. We end this section by comparing the velocity estimates in our Lieb--Robinson bound
to known results for one-soliton solutions of the Toda Lattice. This is done in Section~\ref{subsec:solitons}. 

\subsection{The Toda Lattice}\label{sec:toda}

The Toda Lattice is a well-studied physical model and one of the prototypical discrete integrable wave equations.
We refer to the monographs \cite{fad}, \cite{Te}, \cite{ta} or the review articles \cite{KruegerTeschl1}, \cite{Teschl1}
for further information. Like those discussed in the introduction, the Toda Lattice corresponds to a specific system of
coupled oscillators on $\mathbb{Z}$. We present this model immediately in the infinite volume setting.

For each site $n \in \Z$, we associate an ordered pair $(q_n, p_n) \in \mathbb{R}^2$ and denote by
$\mathcal{X}$ the set of all sequences  ${\rm x}=\{(q_n,p_n)\}_{n\in\Z}$. The system is said to be
comprised of an infinite collection of {\it oscillators}, each situated at a site $n \in \mathbb{Z}$ with position $q_n \in \R$ and 
momentum $p_n \in \R$. Here the state of the system is described by a sequence ${\rm x}=\{(q_n,p_n)\}_{n\in\Z} \in \mathcal{X}$. 

The oscillators evolve in time according to the following coupled system of differential equations. 
For each $n \in \mathbb{Z}$ and any $t \in \mathbb{R}$,
\be \label{eq:qp}
\dot{q}_n(t) = p_n(t)  \quad \mbox{and} \quad \dot{p}_n(t) = e^{-(q_n(t) - q_{n-1}(t))} - e^{-(q_{n+1}(t)-q_n(t))} \, ,
\ee
with initial condition given by some state of the system $\x \in \mathcal{X}$. 
We note that the differential system \eqref{eq:qp} corresponds to the following (formal) 
Hamiltonian $H: \mathcal{X}\rightarrow\R\cup\{\infty\}$ given by 
\be \label{eq:todahampq}
H( {\rm x}) = \sum_{n\in\Z} \fr{p_n^2}{2} +V(q_{n+1}-q_n),
\ee
where $V(r) = e^{-r} + r - 1$. For many sequences, the formal Hamiltonian may be infinite, however, it does formally generate
the system of differential equations in \eqref{eq:qp} through Hamilton's equations. For classes of solutions with prescribed decay
properties we refer to \cite{Teschl2}.

Existence and uniqueness of global solutions to \eqref{eq:qp} on certain subsets of $\mathcal{X}$ is well-known. 
Rather than address this directly, we begin by changing variables.
Suppose that \eqref{eq:qp} has a solution. For each $n \in \Z$ and $t \in \mathbb{R}$, set
\be \label{def:ab}
a_n(t) = \fr{1}{2} e^{-(q_{n+1}(t) -q_n(t))/2} \quad \mbox{and} \quad b_n(t) = - \fr{1}{2} p_n(t) \, ,
\ee
in terms of this given solution.
This choice is commonly referred to as Flaschka variables as they were initially introduced in \cite{Flaschka1}, see also \cite{Flaschka2}. 
Using \eqref{eq:qp}, it is clear that these new variables satisfy the following equations of
motion
\be \label{eq:toda}
\dot{a}_n(t) = a_n(t) \left( b_{n+1}(t) - b_n(t) \right) \quad \mbox{and} \quad 
\dot{b}_n(t) = 2 \left( a_n^2(t) -a_{n-1}^2(t) \right) \, ,
\ee
and correspond to the following formal Hamiltonian
\be \label{eq:todahamab}
H( {\rm x}) = \sum_{n\in\Z} 2b_n^2+ V(- \ln(4a_n^2)) = \sum_{n\in\Z} 2b_n^2+4a_n^2-2\ln{(2a_n)}-1.
\ee
In this work, we will concern ourselves mainly with the properties of the solutions of \eqref{eq:toda}.

Consider the vector space $M= \ell^{\infty}( \mathbb{Z}, \mathbb{R}) \times \ell^{\infty}( \mathbb{Z}, \mathbb{R})$ of pairs of
bounded, real-valued sequences.
We will write each ${\rm x} \in M$ as 
\be
{\rm x} = \left( \{ a_n \}_{n \in \mathbb{Z}} , \{b_n \}_{n \in \mathbb{Z}} \right) = \{ (a_n, b_n) \}_{n \in \mathbb{Z}} \, .
\ee
$M$ is a Banach space with respect to the norm
\be \label{eq:norm}
\| {\rm x } \|_{M} = \max(\|a\|_\infty,\|b\|_\infty), \qquad \|c\|_\infty = \sup_n|c_n| \, .
\ee
Given any initial condition $\x \in M$, global solvability of the system \eqref{eq:toda} is well-known
(see e.g.\ \cite{Te}, Theorem 12.6). 
We will denote the {\it Toda flow} on $M$ by $\Phi_t$, i.e., $\Phi_t : M \to M$ is the function
that associates an initial condition $\x \in M$ to the solution of \eqref{eq:toda} at time $t$:
$\Phi_t(\x) = \{ (a_n(t), b_n(t)) \}_{n \in \mathbb{Z}}$ with $\Phi_0(\x) =\x$.

\subsection{Estimating Toda solutions} \label{subsec:solest}
More is known about the solutions of \eqref{eq:toda} on $M$. In fact, let us fix $\x \in M$ and
introduce operators $L(\x)$, $P(\x): \ell^2(\mathbb{Z}) \rightarrow \ell^2(\mathbb{Z})$ by setting 
\be \label{eq:defL} 
[L(\x)f]_n = a_nf_{n+1} +a_{n-1}f_{n-1} +b_n f_{n} 
\ee 
and
\be \label{eq:defP}
[P(\x)f]_n = a_nf_{n+1} -a_{n-1}f_{n-1} \, .
\ee
For each $\x \in M$, the existence of global solutions imply that 
$L( \Phi_t(\x))$ and $P(\Phi_t(\x))$ are well-defined for all $t \in \mathbb{R}$.
Whenever the initial condition $\x$ is fixed, we will write $L(t) = L( \Phi_t(\x))$ and
$P(t) = P( \Phi_t(\x))$ to spare notation. Observe that for each $\x \in M$, $\Phi_t(\x) \in M$ 
and so $L(t)$ is a bounded self-adjoint operator with the operator norm satisfying
\be \label{eq:Lbd1}
\max\left(\|a(t)\|_\infty,\|b(t)\|_\infty\right)\leq \|L(t)\|_2 \leq 2 \|a(t)\|_\infty+\|b(t)\|_\infty.
\ee
Similarly $P(t)$ is a bounded skew-adjoint operator. 
Since $P(t)$ is also differentiable with respect to $t$, it generates a two-parameter family of unitary propagators $U(t,s)$ 
satisfying
\be \label{eq:uniprop}
\fr{d}{dt}U(t,s) = P(t) U(t,s)  \mbox{ with } U(t,t) = \idty \, ,
\ee
for all pairs $t,s \in \mathbb{R}$, see e.g.\ \cite{Te}, Theorem 12.4 for more details. 
Moreover, a short calculation shows that $P(t)$ and $L(t)$ are a Lax pair associated to \eqref{eq:toda}, i.e.,
\be\label{eq:Lax}
\fr{d}{dt} L(t) = [ P(t), L(t)] \, ,
\ee
and therefore, 
\be
L(t) = U(t,s)L(s)U(t,s)^{-1} \, .\nonumber
\ee
From this fact, each $\x \in M$ satisfies the a-priori estimate
\be \label{toda_unifest}
\| \Phi_t(\x) \|_{M}  \leq \| L(t) \|_2 = \| L(0) \|_2 \quad \mbox{ for all } t \in \mathbb{R}.
\ee
It is important to emphasize the fact that we work with real-valued solutions. In fact, 
while a local existence result can be established in the complex case, even
periodic complex initial conditions can blow up in finite time (see e.g.\ \cite{GHT} or \cite{GHMT}).

Our proof of the Lieb--Robinson bound for the Toda Lattice, see Theorem~\ref{thm:toda_lrb}, makes crucial use of 
the following estimate describing the sensitivity of solutions to changes in the initial condition.
\begin{thm} \label{thm:toda_solest} Let $\x \in M$ and $\mu >0$. There exists a number $v=v(\mu,\x)$
for which given any $n,m \in \mathbb{Z}$, the bound 
\be \label{toda_basest}
\max \left[ \abs{ \fr{\D}{ \D z} a_n(t)}, \abs{ \fr{\D}{ \D z} b_n(t)} \right] \leq 
\fr{8}{\sqrt{17}} \, e^{- \mu \left( |n-m| - v |t| \right)} \, ,
\ee
holds for all $t\in\R$ and each $z\in\{a_m,b_m\}$. In fact, one may take
\be \label{tvel}
v = \left(1+\sqrt{17}\right)\| L(0) \|_2\left(e^{\mu+1}+\fr{1}{\mu}\right).
\ee
\end{thm}

\begin{proof}
Fix $\x \in M$. Global existence of solutions on $M$ guarantees that for each $\x \in M$ and $n \in \mathbb{Z}$, the
function $F_n : \mathbb{R} \to \mathbb{R}^2$ given by
\be
F_n(t; \x)=\begin{pmatrix} a_n(t) \\ b_n(t) \end{pmatrix} \,
\ee
is well-defined. It is differentiable with respect to each $z \in \{ a_m, b_m \}$, e.g.\ as a consequence of Lemma 4.1.9 in \cite{Abraham}.
When convenient, we will suppress the dependence of $F_n$ on $\x$.
Using the equations of motion, i.e.\ \eqref{eq:toda}, it is clear that
\be \label{eq:flowatn}
F_n(t)= F_n(0)+\int_0^t \begin{pmatrix}  a_n(s) \left( b_{n+1}(s) - b_n(s) \right) \\  2 \left( a_n^2(s) -a_{n-1}^2(s) \right) \end{pmatrix} \, ds \, .
\ee
The relation
\be \label{eq:dfnz}
\fr{\D}{ \D z} F_n(t) = \fr{\D}{ \D z} F_n(0) +\sum_{|e|\leq 1}\int_0^t D_{n,e}(s) \fr{\D}{ \D z} F_{n+e}(s) \, ds,
\ee
with 
\be\label{matrix_D_toda}
D_{n,e}(s) = \begin{pmatrix} \left(b_{n+1}(s) -b_n(s) \right)  \delta_0(e) & a_n(s)(-\delta_0(e)+\delta_{1}(e)) \\  4(a_n(s)\delta_0(e) - a_{n-1}(s) \delta_{-1}(e)) & 0  \end{pmatrix} \, .
\ee
then follows immediately from \eqref{eq:flowatn}.

To complete our estimate, we introduce the following notation. For each $v \in \mathbb{R}^2$, we will denote by
\begin{equation}
v = \begin{pmatrix} x \\ y \end{pmatrix} \quad \mbox{and} \quad |v| = \begin{pmatrix} |x| \\ |y| \end{pmatrix} \, .
\end{equation}
Moreover, we will write 
\begin{equation}
\left| \begin{pmatrix} x \\ y \end{pmatrix} \right| \leq \left| \begin{pmatrix} u \\ v \end{pmatrix} \right| \quad \mbox{if and only if} \quad |x| \leq |u| \quad \mbox{and} \quad |y| \leq |v|.
\end{equation}
With this understanding, the uniform solution estimate \eqref{toda_unifest}, and \eqref{matrix_D_toda}, it is clear that
\be \label{eq:1bd}
\left|\fr{\D}{ \D z} F_n(t)\right| \leq \left|\fr{\D}{ \D z} F_n(0)\right| + \| L(0) \|_2 \sum_{|e|\leq 1}\int_0^{|t|} D_{e} \left|\fr{\D}{ \D z} F_{n+e}(s)\right| ds,
\ee
where 
\be \label{eq:matDe}
D_e =  \begin{pmatrix} 2\delta_0(e) & \delta_0(e)+\delta_{1}(e) \\  4(\delta_0(e) + \delta_{-1}(e)) & 0  \end{pmatrix} \, .
\ee
Let us now consider the case that $z = a_m$, i.e.,
\be
\frac{ \partial}{\partial a_m}F_n(0) = \begin{pmatrix} 1 \\ 0 \end{pmatrix} \delta_m(n) \, .
\ee
In this case, iteration of \eqref{eq:1bd} yields
\bea \label{eq:bd2}
\left|\fr{\D}{ \D a_m} F_n(t)\right| &\leq& \sum_{k=0}^\infty  \fr{(\| L(0) \|_2 |t|)^k}{k!} \sum_{|e_1|\leq 1} \cdots \sum_{|e_k| \leq 1} \delta_{m+e_1+\cdots + e_k}(n) D_{e_1}\cdots D_{e_k} \begin{pmatrix} 1 \\ 0 \end{pmatrix} \\ \nonumber
&\le&\sum_{k=|n-m|}^\infty \fr{(\| L(0) \|_2 |t|)^k}{k!} D^k \begin{pmatrix} 1 \\ 0 \end{pmatrix} \, ,
\eea
where we have set
\be
D = \sum_{|e| \leq 1} D_e =  \begin{pmatrix} 2 & 2 \\  8 & 0  \end{pmatrix}.
\ee
Moreover, note that the remainder term in finite iterations of \eqref{eq:1bd} converges to zero
since $\fr{\D}{ \D z} F_n(t)$ is continuous and thus bounded on compact time intervals.

The eigenvalues of $D$ are $\lambda_{\pm} = 1 \pm \sqrt{17}$, and in terms of the eigenvectors $v_{\pm}$ given by
\be
v_{\pm} = \begin{pmatrix} \lambda_{\pm} \\ 8 \end{pmatrix} \quad \mbox{it is clear that} \quad \begin{pmatrix} 1 \\ 0 \end{pmatrix} = \frac{1}{2 \sqrt{17}}v_+ - \frac{1}{2 \sqrt{17}}v_- \, .
\ee
Taking the infinity norm, in $\mathbb{R}^2$, of both sides of \eqref{eq:bd2} shows that
\be \label{toda_L0est}
\left\| \fr{\D}{\D a_m}F_n(t) \right\|_{\infty} \leq \frac{8}{\sqrt{17}} \sum_{k=|n-m|}^\infty\fr{\left( \lambda_+ \| L(0) \|_2 |t|\right)^k}{k!}
\ee

Now, let $\mu>0$ be fixed and set $c = \lambda_+ \| L(0) \|_2$. If $c|t| \leq |n-m| e^{-(\mu+1)}$, then by Stirling
\be \label{eq:stir}
\sum_{k=|n-m|}^{\infty}\fr{(c|t|)^k}{k!} \leq \fr{(c|t|)^{|n-m|}}{(n-m)!} e^{c|t|}
\leq \left(\fr{c|t|}{|n-m|}\right)^{|n-m|}e^{|n-m|}e^{c|t|} \leq e^{-\mu(|n-m|-\fr{c}{\mu}|t|)} \, .
\ee
Otherwise $0 \leq - \mu(|n-m| - ce^{\mu +1} |t|)$, and so by \eqref{toda_L0est} we have
\be
\left\| \fr{\D}{\D a_m}F_n(t) \right\|_{\infty} \leq \frac{8}{\sqrt{17}} e^{c|t|} \leq \frac{8}{\sqrt{17}} e^{-\mu(|n-m|-c(e^{\mu+1}+\fr{1}{\mu})|t|)} \, .
\ee

In the case that $z = b_m$, it is clear that
\be
\begin{pmatrix} 0 \\ 1 \end{pmatrix} = - \frac{ \lambda_-}{8(\lambda_+ - \lambda_-)}v_+ +  \frac{\lambda_+}{8(\lambda_+-\lambda_-)}v_- \, .
\ee 
A short calculation leads to a prefactor of $\frac{\lambda_+}{ \sqrt{17}}$ in the analogue of \eqref{toda_L0est}, and this is less than the one above.
We have proven the result.
\end{proof}

\noindent {\bf Remarks:} 
\vspace{.3cm}
\begin{enumerate}[1.]
\item The number $v$ in our estimate \eqref{toda_basest} depends on the initial condition $\x$ only through 
the quantity $\| L(0)\|_2$. In fact, one can replace $\| L(0) \|_2$ with any uniform estimate on 
the solutions (uniform in $n$ and $t$) as is clear in the bound from \eqref{eq:dfnz} to \eqref{eq:1bd}.
\item If $|n-m|>0$, the left hand side of \eqref{toda_basest} is zero when $t=0$. 
The estimate in \eqref{toda_basest} then shows that the solutions of the Toda Lattice
at site $n$ are insensitive (i.e.\ exponentially small) to changes in the initial condition 
at site $m$ for times $t$ satisfying $v|t| \leq |n-m|$. For this reason, the number
$v$ is often called a bound on the velocity (i.e.\ rate) at which disturbances propagate
through the Toda Lattice. Moreover, if $|n-m| \geq 1$, then the bound in \eqref{toda_L0est}
is at most linear in $|t|$, for small $t$.
\item In general, the fact that solutions of the Toda Lattice have 
partial derivatives that decay exponentially in $|n-m|$ will be more important
than any particular decay rate $\mu$. For this reason, given an
initial condition ${\rm x} \in M$, one can optimize the quantity $v$ over all possible $\mu>0$.
The graph of $f(\mu)=e^{\mu+1}+\fr{1}{\mu}$ is given below.
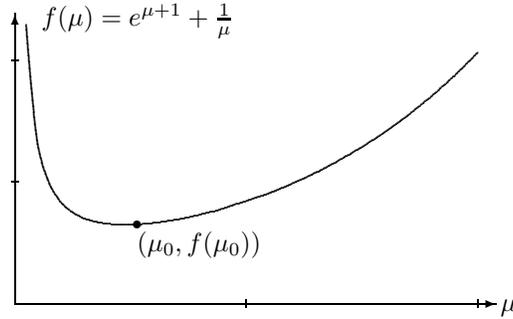
\begin{center}
\begin{picture}(6.5,3.9)

\put(-0.154,0.){\vector(1,0){6.4}}
\put(-0.154,0){\vector(0,1){3.857}}
\put(6.3,-0.1){$\mu$}
\put(0.2,3.7){$f(\mu)=e^{\mu+1}+\fr{1}{\mu}$}

\put(1.47,1.05){\circle*{0.1}}
\put(1.47,0.7){$(\mu_0,f(\mu_0))$}

\put(2.92,-0.05){\line(0,1){0.1}}
\put(6,-0.05){\line(0,1){0.1}}
\put(-0.20,1.62){\line(1,0){0.1}}
\put(-0.20,3.24){\line(1,0){0.1}}

\curve(0.,3.708, 0.15,2.13, 0.3,1.611, 0.45,1.363, 0.6,1.226,
0.75,1.144, 0.9,1.095, 1.05,1.067, 1.2,1.053, 1.35,1.051, 1.5,1.057,
1.65,1.069, 1.8,1.088, 1.95,1.111, 2.1,1.139, 2.25,1.171, 2.4,1.207,
2.55,1.246, 3.15,1.442, 3.45,1.561, 3.75,1.696, 4.05,1.848,
4.5,2.108, 4.8,2.307, 5.1,2.527, 5.4,2.771, 5.7,3.041, 6.,3.34)
\end{picture}
\end{center}
It is clear that the optimal $\mu$ is achieved when 
\be \label{eq:lamw}
\fr{d}{d\mu}\left(e^{\mu+1}+\fr{1}{\mu}\right)=0 \, 
\ee
and moreover, this value is independent of the initial condition $\x$.
The unique solution of \eqref{eq:lamw} can be expressed in terms of the Lambert $W$-function (\cite[\S4.13]{dlmf})
and is given by
\be
\mu_0= 2 W(1/(2\sqrt{e})) \approx 0.47767 \quad \mbox{and} \quad f(\mu_0) \approx 6.47622 \, .
\ee
\end{enumerate}
\subsection{A Lieb--Robinson bound for Toda}  \label{subsec:lrbtoda}
We can now formulate a Lieb--Robinson bound for the Toda Lattice. In general, an {\it observable} 
$A$ is a function $A : M \to \C$. Let us denote by $\A$ the set of all observables with
well-defined first order partial derivatives. We will say that an observable $A$ is {\it pointwise bounded} if
\be\label{defAinf}
\| A \|_{\x} = \sup_{t\in\R} |A(\Phi_t(\x))| < \infty \quad \forall\, \x\in M\, .
\ee
Take $\mathcal{A}^{(1)}$ to be the set of those observables in $\A$ for which all of the
first order partial derivatives are bounded and
\be
\| A \|_{1, \x} = \sum_{n \in \mathbb{Z}}\left( \left\|\fr{\D A}{\D a_n} \right\|_\x +  \left\|\fr{\D A}{\D b_n} \right\|_\x \right)< \infty
\quad \forall\, \x\in M\, .
\ee
An observable $A \in \A$ is said to be {\it supported} on $X \subset \mathbb{Z}$ if the observables
$\fr{ \D A}{ \D a_n}$ and $\fr{ \D A}{ \D b_n}$ are identically zero for all $n \in \mathbb{Z} \setminus X$.
The {\it support} of an observable $A$ is the minimal set on which $A$ is supported, and we will
denote this set by ${\rm supp}(A)$. Set $\A_0 \subset \A$ to be the set of all observables with
compact support. For any $t \in \mathbb{R}$ and $A \in \A$, the {\it Toda dynamics}, which we denote by
$\alpha_t$, is the observable-valued mapping given by
\be
\alpha_t(A) = A \circ \Phi_t \, ,
\ee
where $\Phi_t$ is the Toda flow described at the end of Section~\ref{sec:toda} above. In other words, for each ${\rm x} \in M$, $t \in \mathbb{R}$, and $A \in \A$,
\be
[ \alpha_t(A) ]( {\rm x}) = A( \Phi_t({\rm x}) ) \, ,
\ee  
which, since solutions are global on $M$, is a well-defined quantity.
\begin{Example} \label{ex:1} The two most basic observables correspond to evaluation maps, i.e.,
to each $n \in \mathbb{Z}$, we associate the functions $A_n$ and $B_n$ given by  
\be
A_n( {\rm x}) = a_n \quad \mbox{and} \quad B_n({\rm x}) = b_n \quad \mbox{for each} \quad {\rm x} \in M \, .
\ee
Each of these observables have support $X = \{ n\}$. Moreover, the dynamically evolved observables
\be
[ \alpha_t(A_n) ]({\rm x}) = a_n(t) \quad \mbox{and} \quad [ \alpha_t(B_n) ]({\rm x}) = b_n(t)
\ee
correspond to observations over time.
\end{Example}

The {\it modified Poisson bracket} of two observables $A$ and $B$ is formally defined as 
the observable
\be \label{eq:mpb}
\left[ \{A,B\} \right] (\x)  = \fr{1}{4}\sum_{n\in\Z} a_n \cdot \left[ \fr{\D A}{\D a_n}  \cdot \fr{\D B}{\D \tilde{b}_n} -\fr{\D  A}{\D \tilde{b}_n} \cdot \fr{\D B}{\D a_n} \right](\x)  
\ee
where we have denote by $\fr{\D}{\D \tilde{b}_n} = \fr{\D}{\D b_{n+1}} - \fr{\D}{\D b_n}$. 
If, for example, either $A$ or $B$ has compact support, then the corresponding modified Poisson bracket
is a well-defined observable. This quantity is of particular interest since it generates the Toda dynamics, i.e.,
for any $A \in \A_0$ 
\be\label{eq:eompobr}
\fr{d}{dt} \alpha_t(A) = \alpha_t \left( \left\{ A, H \right\} \right) = \left\{ \alpha_t(A), H \right\} \, ,
\ee
where $H$ is the Hamiltonian \eqref{eq:todahamab}.
Now, for any $n, m \in \mathbb{Z}$, $t \in \mathbb{R}$, and ${\rm x} \in M$, it is easy to 
see that
\be
\left\{ \alpha_t(A_n), B_m \right\}({\rm x}) = \fr{1}{4} a_{m-1} \fr{\D }{\D a_{m-1}}a_n(t) - \fr{1}{4} a_{m} \fr{\D }{\D a_{m}}a_n(t) \, 
\ee
using the basic observables defined in Example~\ref{ex:1}. In this simple case, Theorem~\ref{thm:toda_solest} gives
the bound
\begin{eqnarray}
\abs{ \left\{ \alpha_t(A_n), B_m \right\}({\rm x})} & \leq & \fr{2 \| a \|_{\infty}}{ \sqrt{17}}  \left( e^{- \mu (|n-m+1| - v|t|)} +  e^{- \mu (|n-m| - v|t|)} \right) \nonumber \\
&  \leq &\fr{2 \| a \|_{\infty}}{\sqrt{17}}  (1+ e^{\mu}) e^{- \mu(|n-m|-v|t|)} \, .
\end{eqnarray}
In general, we have the following Lieb--Robinson bound for the Toda Lattice.
\begin{thm} \label{thm:toda_lrb} Let $\x \in M$ and $\mu >0$. There exist numbers $C$ and $v$ for which given 
any observables $A, B \in \A^{(1)}$, the estimate 
\be \label{ineq:toda_gentbd}
|\{\alpha_t(A),B\}({\rm x})| \leq C\, \| a \|_{\infty} \sum_{n, m \in \mathbb{Z}} \left( \left\| \fr{\D A }{ \D a_m}  \right\|_\x + \left\| \fr{\D A}{ \D b_m}  \right\|_\x \right) \left( \left\| \fr{\D B}{ \D a_n}  \right\|_\x + \left\| \fr{\D B}{ \D b_n}  \right\|_\x \right) e^{- \mu \left( |n-m| -v|t| \right)} 
\ee
holds for all $t\in\R$. Here  $C = \frac{2}{\sqrt{17}}(1+e^{\mu}) $ and $v$ is as in Theorem~\ref{thm:toda_solest}.
\end{thm}

\begin{proof}
Substituting into \eqref{eq:mpb}, we find that
\be
\left\{ \alpha_t(A), B \right\} ({\rm x})  =  \fr{1}{4} \sum_{n \in \mathbb{Z}} a_n \cdot \left[ \fr{\D}{\D a_n} \alpha_t(A) \cdot \fr{\D B}{\D \tilde{b}_n} -\fr{\D }{\D \tilde{b}_n} \alpha_t(A) \cdot \fr{\D B}{\D a_n} \right]({\rm x}) \, .
\ee
Using the chain rule, it is clear that for any $z \in \{a_n, \tilde{b}_n\}$,
\be
\left[ \fr{\D}{ \D z} \alpha_t(A) \right] ({\rm x}) = \sum_{m \in \mathbb{Z}} \left[ \alpha_t \left( \fr{\D A}{ \D a_m}\right) \cdot \fr{\D}{ \D z} \alpha_t(A_m) +
\alpha_t \left( \fr{\D A}{ \D b_m} \right) \cdot \fr{\D}{ \D z} \alpha_t(B_m) \right]({\rm x})
\ee
where we are using the notation from Example~\ref{ex:1}. 
Applying the triangle inequality and then Theorem~\ref{thm:toda_solest}, the bound
\bea
|\{\alpha_t(A),B\}({\rm x})| & \leq &  \fr{\| a \|_{\infty}}{4} \sum_{n, m} \left( \left\| \fr{\D A}{ \D a_m}  \right\|_\x \left| \fr{ \D}{ \D a_n} a_m(t) \right| + \left\| \fr{\D A}{ \D b_m}  \right\|_\x \left| \fr{ \D}{ \D a_n} b_m(t) \right| \right) \left\| \fr{\D B}{ \D \tilde{b}_n} \right\|_\x \nonumber \\ 
& \mbox{ } & +  \fr{\| a \|_{\infty}}{4} \sum_{n, m} \left(  \left\| \fr{\D A}{ \D a_m}  \right\|_\x \left| \fr{ \D}{ \D \tilde{b}_n} a_m(t) \right|  +  
\left\| \fr{\D A}{ \D b_m} \right\|_\x \left| \fr{ \D}{ \D \tilde{b}_n} b_m(t)\right|  \right) \left\| \fr{\D B}{ \D a_n}  \right\|_\x \nonumber \\
& \leq &  \fr{2 \| a \|_{\infty}}{ \sqrt{17}} \sum_{n, m} \left( \left\| \fr{\D A}{ \D a_m}  \right\|_\x + \left\| \fr{\D A}{ \D b_m} \right\|_\x \right) \left\| \fr{\D B}{ \D \tilde{b}_n} \right\|_\x e^{-\mu \left(|n-m| - v|t|\right)} \nonumber \\ 
& \mbox{ } & +  \fr{2 \| a \|_{\infty}}{ \sqrt{17}} e^{\mu} \sum_{n, m} \left(  \left\| \fr{\D A}{ \D a_m}  \right\|_\x  +  
\left\| \fr{\D A}{ \D b_m}  \right\|_\x \right) \left\| \fr{\D B}{ \D a_n} \right\|_\x e^{-\mu \left(|n-m| - v|t| \right)} 
\eea
readily follows. Note that the bound on e.g.\ $\left| \fr{ \D}{ \D \tilde{b}_n} a_m(t) \right|$ follows from the argument in Theorem~\ref{thm:toda_solest}. 
Another triangle inequality proves \eqref{ineq:toda_gentbd} as claimed.
\end{proof}

For many applications, the observables of interest will have disjoint supports. 
In this case, the bound in \eqref{ineq:toda_gentbd} can be stated as follows. 
For any ${\rm x} \in M$ and $\mu >0$, there exist numbers $C$ and $v$, 
as above, such that given any two disjoint subsets $X$ and $Y$ of
$\mathbb{Z}$ the bound 
\be \label{eq:tlrb2}
|\{\alpha_t(A),B\}({\rm x})| \leq C \| a \|_{\infty} \| A \|_{1, \x} \| B \|_{1, \x} e^{- \mu(d(X,Y) - v|t|)}  \, ,
\ee
holds for all $t \in \mathbb{R}$ and any observables $A,B \in \A^{(1)}$ with supports in $X, Y$,
respectively. Here
\be
d(X,Y) = \inf \{ |x-y| \, : \, x \in X \mbox{ and } y \in Y \} >0 \, ,
\ee
and \eqref{eq:tlrb2} corresponds with the bound \eqref{eq:genhamlrb} claimed in the introduction.
Other norms on the observables under considerations are often useful.
For example, consider
\be
\| \D A \|_\x = \sup_{m \in \mathbb{Z}} \max \left[ \left\| \fr{\D A}{ \D a_m} \right\|_\x , \left\| \fr{\D A}{ \D b_m} \right\|_\x \right]
\ee
for any $A \in \A^{(1)}$. Theorem~\ref{thm:toda_lrb} then shows that: For any ${\rm x} \in M$ and 
$\mu >0$, there exist numbers $C$ and $v$, as above, such that given any two 
disjoint subsets $X$ and $Y$ of
$\mathbb{Z}$, at least one of which being finite, the bound 
\be
|\{\alpha_t(A),B\}({\rm x})| \leq 4 C \| a \|_{\infty} \| \D A \|_\x \| \D B \|_\x \sum_{m \in X, n \in Y} e^{- \mu(|m-n| - v|t|)}  \, ,
\ee
holds for all $t \in \mathbb{R}$ and any observables $A,B \in \A^{(1)}$ with supports in $X, Y$,
respectively. 

In any case, the bound in Theorem~\ref{thm:toda_lrb} demonstrates that each $\x \in M$ propagates no faster than some
finite rate. We say that the Lieb--Robinson velocity corresponding to $\x \in M$ is $v(\x, \mu_0)$ where $\mu_0$ is
the infimum defined in the remarks after Theorem~\ref{thm:toda_solest}.  

\subsection{On one-soliton solutions} \label{subsec:solitons}
Perhaps the most important collection of solutions to the Toda Lattice are the solitons, i.e.\ the solitary waves, see e.g.\ \cite{Te,Teschl1,ta}. In fact, solitons can be considered as the stable part of any short-range initial condition since
every such solution eventually splits into a number of stable solitons plus a decaying dispersive tail \cite{KruegerTeschl1}.
Since all involved quantities can be computed explicitly for the one-soliton solution we will use it as a test case and
for our Lieb--Robinson bounds. 

Fix $\kappa > 0$. A one-soliton solution of the Toda lattice is given by
\begin{equation}
q_n(t; \pm) = q - \ln \left( \frac{1+ \exp \left[ - 2 \kappa n \pm 2 \sinh(\kappa) t + \delta\right] }{1+ \exp \left[ - 2 \kappa (n-1) \pm 2 \sinh(\kappa) t +\delta\right]} \right) \, ,
\end{equation}
where $q$ and $\delta$ are real constants. Here $q_n(t; \pm)$ represents the position of the traveling wave.
It describes a single bump traveling with speed $\pm \frac{\sinh(\kappa)}{ \kappa}$ and width proportional to $1/ \kappa$.
In other words, the smaller the soliton the faster it propagates. Changing $\delta$ amounts to a shift of the solution and
we will set $\delta=0$ for simplicity. {F}rom  this explicit formula it is clear that, in contrast to harmonic models, the velocity
of a solution to the Toda Lattice may indeed depend on the initial condition.

In terms of the function
\begin{equation}
f_{\pm}(x,t) = 1 + \exp \left[ - 2 \kappa x \pm 2 \sinh(\kappa) t \right]
\end{equation}
it is clear that
\begin{equation}
q_n(t; \pm) = q - \ln \left( \frac{f_{\pm}(n,t)}{f_{\pm}(n-1,t)} \right) \, ,
\end{equation}
and from Hamilton's equations, we also know that
\begin{equation}
p_n(t; \pm) = \frac{d}{dt} q_n(t; \pm) = \frac{f_{\pm}'(n-1,t)}{f_{\pm}(n-1,t)} - \frac{f_{\pm}'(n,t)}{f_{\pm}(n,t)} \, .
\end{equation}
In Flaschka's variables, we have that
\begin{equation}
a_n(t; \pm) = \frac{1}{2} \frac{\sqrt{f_{\pm}(n-1,t) f_{\pm}(n+1,t)}}{f_{\pm}(n,t)} \quad \mbox{and} \quad b_n(t; \pm) = \frac{1}{2} \left( \frac{f_{\pm}'(n,t)}{f_{\pm}(n,t)} - \frac{f_{\pm}'(n-1,t)}{f_{\pm}(n-1,t)} \right) \, .
\end{equation}
A short calculation shows that
\begin{equation}
\sup_na_n(0; \pm) = a_0(0; \pm) = \frac{\cosh(\kappa)}{2} \quad \mbox{and} \quad \sup_n|b_n(0; \pm)| = |b_0(0; \pm)| = \frac{\sinh(\kappa)\tanh(\kappa)}{2} \, .
\end{equation}
With $\x_{\kappa}$ denoting the initial condition corresponding to this one-soliton, we have proven that
\begin{equation}
\frac{\cosh(\kappa)}{2} = \max\left(\frac{\cosh(\kappa)}{2}, \frac{\sinh(\kappa)\tanh(\kappa)}{2}\right) \leq \| L( \x_{\kappa}) \|_2 \leq  \cosh(\kappa) +\frac{\sinh(\kappa)\tanh(\kappa)}{2} \, ,
\end{equation}
using e.g.\ \eqref{eq:Lbd1}. Since $L(\x_{\kappa})$ is self-adjoint its norm is equal to the spectral radius. Moreover,
the spectrum is given by an absolutely continuous part $[-1,1]$ plus the single eigenvalue $\pm\cosh(\kappa)$ implying
\begin{equation}
 \| L( \x_{\kappa}) \|_2 = \cosh( \kappa) \, .
\end{equation}
To see this last claim note that the one-soliton solution can be computed from the inverse scattering transform by choosing one eigenvalue $\lambda=\pm\cosh(\kappa)$ plus the corresponding norming constant $\gamma=(1-\mathrm{e}^{-2\kappa})\mathrm{e}^\delta$ and zero reflection coefficient (cf.\ \cite[Sect.~3.6]{ta} or
\cite[Sec.~13.4]{Te}) or by using the double commutation method to add one eigenvalue $\lambda$ with norming
constant $\gamma$ to the trivial solution (cf.\ \cite[Sect.~14.5]{Te}).

As is clear from this calculation, the Lieb--Robinson velocity does provide a reasonable estimate on
the actual velocity, at least for one-soliton solutions. 

%
%
%
\section{Estimates on Perturbed Toda Systems}\label{sec:pert_toda}
In this section, we consider a class of perturbations of the Toda system for which locality results
analogous to Theorem~\ref{thm:toda_lrb} still hold.  
We introduce these perturbations as follows. Let $W : \mathbb{R} \to [0,\infty)$ satisfy $W \in C^2( \mathbb{R})$.
Consider the formal Hamiltonian
\be \label{eq:pertham}
H^\w = H + \sum_{n \in \mathbb{Z}} W_n,
\ee
where $H$ is the Toda Hamiltonian, see \eqref{eq:todahamab}, and for any initial condition
${\rm x}=\{(a_n,b_n)\}_{n\in\Z} \in M_0 = \ell^{\infty}( \mathbb{Z}, \mathbb{R}\setminus\{0\}) \times \ell^{\infty}( \mathbb{Z}, \mathbb{R})$, the observable $W_n(\x)=W( \ln(4a_n^2))$. Our choice of 
parametrization is motivated by the fact that the potential $V$ for the unperturbed Toda Lattice, 
see \eqref{eq:todahampq}, is a function of $q_{n+1}-q_n = - \ln(4a_n^2)$. Since the choice $a_n = 0$
corresponds to taking the positions of the particles at sites $n$ and $n+1$
infinitely far apart, we will exclude it from our considerations. The formal Hamiltonian \eqref{eq:pertham}
corresponds to the following system of coupled differential equations
\begin{eqnarray} \label{eq:pertevo}
\dot{a}_n^{\w}(t) &= &a_n^{\w}(t) \left( b_{n+1}^{\w}(t) - b_n^{\w}(t) \right)\\
\dot{b}_n^{\w}(t) &= & 2 \left( a_n^{\w}(t)^2 -a_{n-1}^{\w}(t)^2 \right)+R_n(t), \nonumber
\end{eqnarray}
where 
\be
R_n(t) = \fr{1}{2}\left[W'( \ln(4a_n^{\w}(t)^2)) - W'( \ln(4a_{n-1}^{\w}(t)^2))\right].
\ee
Local existence and uniqueness of solutions of \eqref{eq:pertevo} corresponding to initial conditions $\x \in M_0$
follows from standard results, \cite[Thm.~4.1.5]{Abraham}.
Let us denote by $\Phi^{\w}_t$ the perturbed Toda flow, i.e., the function that associates initial conditions
$\x \in M_0$ with $\Phi^{\w}_t(\x) = \left\{ \left(a_n^{\w}(t), b_n^{\w}(t) \right) \right\}$, the solution of \eqref{eq:pertevo} at time $t$.
Consider the set of initial conditions $M_b =M_b(W) \subset M_0$ for which there exists numbers $C_1, C_2 < \infty$ with
\be\label{def:C1C2}
\sup_{t \in \mathbb{R}} \| \Phi^{\w}_t( \x) \|_M \leq C_1 \quad \mbox{and} \quad \sup_{t \in \mathbb{R}} \sup_{n \in \mathbb{Z}} \frac{1}{|a_n^{\w}(t)|} \leq C_2 \, .
\ee
It will be shown in Appendix~\ref{sec:app} that this $M_b$ contains at least all initial conditions whose energy is finite
under appropriate assumptions on $W$.

For initial conditions $\x \in M_b$, we have two estimates on the corresponding Lieb--Robinson velocity. 
The first is obtained in Section~\ref{subsec:direct} by simply arguing as we did in Theorem~\ref{thm:toda_solest}. 
The other, in Section~\ref{subsec:inter}, achieves an estimate of the perturbed velocity (corresponding to $\x$)
in terms of the unperturbed velocity (corresponding to $\x$) via interpolation.
%
%
\subsection{Direct Bounds} \label{subsec:direct}
Following closely the arguments in Theorem~\ref{thm:toda_solest}, we obtain the next result.
\begin{thm} \label{thm:pert_toda_solest}
Fix $W \in C^2( \mathbb{R})$ with $W'' \in L^{\infty}(\mathbb{R})$, $\mu >0$, and let $\x \in M_b$. There exist numbers $C^{\w} = C^{\w}( \x, W)$ and $v^{\w} = v^{\w}( \mu, \x)$ for which given any 
$n,m\in\Z$, the bound
\be
\max \left[ \abs{ \fr{\D}{ \D z} a_n^\w(t)}, \abs{ \fr{\D}{ \D z} b_n^\w(t)} \right] \leq C^{\w} e^{-\mu(|n-m|-v^{\w}|t|)},
\ee
holds for all $t\in\R$ and each $z\in\{a_m,b_m\}$. In fact, one may take
\be
v^{\w} = \left(1+\sqrt{17 + \frac{4C_2 \| W'' \|_{\infty}}{C_1}} \right)C_1\left(e^{\mu+1}+\fr{1}{\mu}\right)
\ee
with $C_1, C_2$ from \eqref{def:C1C2}.
\end{thm}
\begin{proof}
Fix $\x \in M_b$. Global existence again guarantees that for each $n \in \mathbb{Z}$, the function
\be
F^\w_n(t) = \begin{pmatrix} a_n^{\w}(t) \\ b_n^{\w}(t) \end{pmatrix}
\ee
satisfies
\be
F^{\w}_n(t) = F^{\w}_n(0)  + \int_0^t \begin{pmatrix} a_n^{\w}(s) \left( b_{n+1}^{\w}(s) - b_n^{\w}(s) \right) \\ 2 \left( a_n^{\w}(s)^2 -a_{n-1}^{\w}(s)^2 \right)+R_n(s) \end{pmatrix}ds. 
\ee
Since $W$ is sufficiently smooth, the components of $F^{\w}_n(t)$ are differentiable with respect to each $z \in \{ a_m, b_m \}$, and
the bound 
\be
\left|\fr{\D}{ \D z} F^{\w}_n(t)\right| \leq \left|\fr{\D}{ \D z} F^{\w}_n(0)\right| + \sum_{|e|\leq 1}\int_0^{|t|} D^{\w}_{e} \left|\fr{\D}{ \D z} F^{\w}_{n+e}(s)\right| ds,
\ee
follows as in \eqref{eq:1bd} with 
\be
D^{\w}_e =  \begin{pmatrix} 2C_1\delta_0(e) & C_1 \left(\delta_0(e)+\delta_{1}(e) \right) \\  \left(4C_1 + C_2 \| W'' \|_{\infty} \right) \left( \delta_0(e) + \delta_{-1}(e) \right) & 0  \end{pmatrix} \, .
\ee
Following the previous scheme, iteration (with $z = a_m$) yields
\be
\left|\fr{\D}{ \D a_m} F^{\w}_n(t)\right| \leq \sum_{k =|n-m|}^{\infty} \frac{(2C_1 |t|)^k}{k!} (D^{\w})^k \begin{pmatrix} 1 \\ 0  \end{pmatrix}
\ee
with
\be
D^\w = \begin{pmatrix} 1 & 1 \\ \alpha & 0  \end{pmatrix} \quad \mbox{for} \quad \alpha = 4 + \frac{C_2 \| W'' \|_{\infty}}{C_1} \, .
\ee
Eigenvalues and eigenvectors of $D^{\w}$ are 
\be
\lambda_{\pm} = \frac{1 \pm \sqrt{1+4 \alpha}}{2} \quad \mbox{and} \quad v_{\pm} = \begin{pmatrix} \lambda_{\pm} \\ \alpha \end{pmatrix}
\ee
and arguing as before, it is now clear that
\be
\left\| \fr{\D}{\D a_m}F^{\w}_n(t) \right\|_{\infty} \leq \frac{2 \alpha}{ \sqrt{1+4 \alpha}}  e^{-\mu(|n-m|-v^{\w}|t|)}
\ee
with
\be 
v^{\w}= 2 \lambda_+ C_1\left(e^{\mu+1}+\fr{1}{\mu}\right) \, .
\ee
For $z=b_m$, one similarly finds 
\be
\left\| \fr{\D}{\D b_m}F^{\w}_n(t) \right\|_{\infty} \leq \frac{2 \lambda_+ }{ \sqrt{1+4 \alpha}}  e^{-\mu(|n-m|-v^{\w}|t|)}
\ee
This completes the proof.
\end{proof}
%
%

\subsection{Bounds Via Interpolation} \label{subsec:inter}
The goal of this section is to prove a different bound on the Lieb--Robinson velocity corresponding to an initial condition $\x \in M_b$, see Theorem~\ref{thm:pert_toda_interpol} below.
The novel feature of this estimate is that it is explicit in the unperturbed Lieb--Robinson velocity of $\x$. 

Before we state the main result of this section, we first indicate some further estimates on
the unperturbed system which will be useful in proving Theorem~\ref{thm:pert_toda_interpol}.
To start with, we prove another Lieb--Robinson type estimate for the Toda system, see Lemma~\ref{lem2} below. 
Afterwards, we introduce a quantity that is better suited for
the iteration scheme which is at the heart of proving Theorem~\ref{thm:pert_toda_interpol}, see $G_{\mu}(k)$ in \eqref{eq:defgmu} below.
Lastly, we will state and prove Theorem~\ref{thm:pert_toda_interpol}.

We begin with the following lemma on second order derivatives of the unperturbed system. 
\begin{lem}\label{lem2}
Fix $\mu >0$ and let $\x \in M$. There exists a number $C = C( \mu, \x)>0$  and a function $h$, depending on $\mu$ and $\x$, for which
given any $n,k, \ell \in\Z$, the estimate
\be
\max \left[ \left| \fr{\D^2}{ \D z \D \tilde{b}_k} a_n(t) \right| , \left| \fr{\D^2}{ \D z \D \tilde{b}_k} b_n(t) \right| \right] \leq C e^{- \mu |n- \ell|} e^{- \mu |n-k|} e^{2 \mu v |t|} h(t)
\ee
holds for all $t\in\R$. Here $a_n(t)$ and $b_n(t)$ are the solutions of \eqref{eq:toda} with initial condition $\x \in M$,
$\fr{\D}{\D \tilde{b}_k} = \fr{\D}{\D b_{k+1}} -\fr{\D}{\D b_k}$, $z \in \{a_\ell , b_\ell \}$, and the number $v$ is as in Theorem \ref{thm:toda_solest}.
The function $h$ grows at most exponentially. 
\end{lem}
\begin{proof}
Fix $\x \in M$. As in Theorem~\ref{thm:toda_solest}, it is clear that the function $F_n(t)$ is
well-defined for each $n \in \mathbb{Z}$ and $t \in \mathbb{R}$. In addition, both
\eqref{eq:dfnz} and \eqref{matrix_D_toda} still hold with the choice $z = \tilde{b}_k$. Taking a second
derivative, we find that
\be
\fr{\D^2}{ \D z \D \tilde{b}_k} F_n(t) = 
\sum_{|e|\leq 1}\int_0^t \left( \fr{\D}{ \D z} D_{n,e}(s) \fr{\D}{ \D \tilde{b}_k} F_{n+e}(s) + D_{n,e}(s) \fr{\D^2}{\D z \D \tilde{b}_k} F_{n+e}(s) \right) \, ds ,
\ee
for each $z \in \{ a_\ell, b_\ell \}$, since $\fr{\D^2}{ \D z \D \tilde{b}_k} F_n(0)  = 0$. 

The bound
\bea \label{eq:2derbd1}
\left| \fr{\D^2}{ \D z \D \tilde{b}_k} F_n(t) \right| & \leq & \frac{8}{\sqrt{17}} e^{-\mu|n- \ell|} \sum_{|e|\leq 1}\int_0^{|t|} e^{\mu v s}  D_{e}' \left| \fr{\D}{ \D \tilde{b}_k} F_{n+e}(s) \right| \, ds + \\
& \mbox{ } & \quad + \| L(0) \|_2 \sum_{|e| \leq 1} \int_0^{|t|} D_{e} \left| \fr{\D^2}{\D z \D \tilde{b}_k} F_{n+e}(s) \right|  \, ds  \nonumber 
\eea
with
\be
D_e' = \begin{pmatrix} (e^{\mu}+1) \delta_0(e) & \delta_0(e) + \delta_{1}(e) \\ 4 \left( \delta_0(e) + e^{\mu} \delta_{-1}(e) \right) & 0 \end{pmatrix}
\ee
and $D_e$ as in \eqref{eq:matDe}, follows using Theorem~\ref{thm:toda_solest} and the argument therein. A further application of 
Theorem~\ref{thm:toda_solest} implies that
\be
\left| \fr{\D}{ \D \tilde{b}_k} F_{n+e}(s) \right| \leq \frac{8}{\sqrt{17}} e^{\mu} e^{-\mu|n+e-k|} e^{\mu v s} \begin{pmatrix} 1 \\ 1 \end{pmatrix}
\ee
and therefore the first term on the right hand side of \eqref{eq:2derbd1} can be estimated by
\be
\frac{64}{17} e^{\mu} e^{- \mu |n-\ell|} e^{-\mu |n-k|} \int_0^{|t|} e^{2 \mu v s} \, ds \begin{pmatrix} 2(e^{\mu}+1) \\ 4(e^{2\mu} +1) \end{pmatrix}  \, .
\ee
Here we have used that
\be
\sum_{|e| \leq 1} e^{-\mu|n+e-k|} D_e' \begin{pmatrix} 1 \\ 1 \end{pmatrix} \leq e^{-\mu|n-k|} \begin{pmatrix} 2(e^{\mu}+1) \\ 4(e^{2\mu} +1) \end{pmatrix} \, .
\ee

Let us introduce the notation
\be
C_{\mu} = \frac{64}{17} e^{\mu} \quad \mbox{and} \quad y = \begin{pmatrix} 2(e^{\mu}+1) \\ 4(e^{2\mu} +1) \end{pmatrix} \, .
\ee
We have proven that
\be \label{eq:2derbd2}
\left| \fr{\D^2}{ \D z \D \tilde{b}_k} F_n(t) \right|  \leq  C_{\mu} e^{- \mu |n-\ell|} e^{-\mu |n-k|} \int_0^{|t|} e^{2 \mu v s} \, ds \cdot y +  \| L(0) \|_2 \int_0^{|t|}  \sum_{|e| \leq 1} D_{e} \left| \fr{\D^2}{\D z \D \tilde{b}_k} F_{n+e}(s) \right|  \, ds  
\ee
Iteration now yields
\bea \label{eq:2derbd3}
\left| \fr{\D^2}{ \D z \D \tilde{b}_k} F_n(t) \right|  & \leq &   C_{\mu} \sum_{j=0}^{\infty} \| L(0) \|_2^j \int_0^{|t|} \int_0^{t_1} \cdots \int_0^{t_j} e^{2 \mu v t_{j+1}} d t_{j+1} \cdots d t_1 \times \nonumber \\
& \mbox{ } & \quad \times \sum_{|e_1| \leq 1} \cdots \sum_{|e_j| \leq 1} e^{- \mu |n+e_1 + \cdots +e_j-\ell|} e^{-\mu |n+e_1 + \cdots +e_j-k|} 
D_{e_1} D_{e_2} \cdots D_{e_j} y 
\eea
and the convergence is guaranteed as before.

A short calculation shows that
\be \label{eq:intvel}
 \int_0^{|t|} \int_0^{t_1} \cdots \int_0^{t_j} e^{2 \mu v t_{j+1}} \, d t_{j+1} \cdots d t_1= \frac{1}{(2 \mu v)^{j+1}} \sum_{m=j+1}^{\infty} \frac{(2 \mu v |t|)^m}{m!}
\ee
whereas the bound $e^{- \mu|z+e|} \leq e^{\mu |e|} e^{- \mu |z|}$ immediately implies that
\be
e^{-\mu|n+e_1+e_2 + \cdots +e_j - \ell|} \leq e^{\mu|e_1|} \cdots e^{\mu|e_j|} e^{-\mu|n-\ell|} \, .
\ee
This proves that
\be \label{eq:2derbd4}
\left| \fr{\D^2}{ \D z \D \tilde{b}_k} F_n(t) \right|   \leq   \frac{C_{\mu}}{2 \mu v} e^{- \mu|n-\ell|} e^{-\mu|n-k|} \sum_{j=0}^{\infty} \left( \frac{\| L(0) \|_2}{2 \mu v} \right)^j 
\sum_{m=j+1}^{\infty} \frac{(2 \mu v |t|)^m}{m!} \left( \tilde{D} \right)^j y
\ee
where
\be
\tilde{D} = \sum_{|e| \leq 1} e^{2 \mu |e|} D_e =  \begin{pmatrix} 2 &e^{2 \mu} +1 
\\ 4(e^{2 \mu}+1) & 0 \end{pmatrix} \, .
\ee
The eigenvalues of $\tilde{D}$ are $\lambda_{\pm} = 1 \pm \sqrt{1+ 4(e^{2\mu}+1)^2}$ and a convenient choice of eigenvectors are
\be
v_{\pm} = \begin{pmatrix} \lambda_{\pm} \\ 4(e^{2 \mu}+1) \end{pmatrix} \, .
\ee
In terms of these, it is clear that
\be
y = y_1 v_+ + y_2 v_- \quad \mbox{with} \quad y_1 = \frac{2(e^{\mu}+1) -  \lambda_-}{\lambda_+ - \lambda_-} \quad \mbox{and} \quad y_2 = \frac{\lambda_+ -2(e^{\mu}+1)}{\lambda_+-\lambda_-} \, .
\ee
Inserting this into \eqref{eq:2derbd4}, we find that
\bea
\left\| \fr{\D^2}{ \D z \D \tilde{b}_k} F_n(t) \right\|_{\infty} & \leq & 
\frac{C_{\mu}}{2 \mu v}  e^{-\mu |n- \ell|} e^{-\mu |n-k|} \sum_{j=0}^{\infty} \left( \frac{\| L(0) \|_2}{2 \mu v} \right)^j \sum_{m=j+1}^{\infty} \frac{(2 \mu v|t|)^m}{m!} \left\| y_1 \lambda_+^j v_+ + y_2 \lambda_-^j v_- \right\|_{\infty}  \nonumber \\
& \leq & C e^{-\mu |n- \ell|} e^{-\mu |n-k|} \sum_{j=0}^{\infty}\left( \frac{\| L(0) \|_2 \lambda_+}{2 \mu v} \right)^j \sum_{m=j+1}^{\infty} \frac{(2 \mu v|t|)^m}{m!}
\eea
It is easy to see that for any $A,B \in \mathbb{R}$,
\be \label{eq:bigsum}
\sum_{j=0}^{\infty} A^j \sum_{m=j+1}^{\infty} \frac{(B|t|)^m}{m!} =
\sum_{m=1}^{\infty} \frac{(B|t|)^m}{m!} \sum_{j=0}^{m-1} A^j =
\begin{cases} \frac{1}{A-1} \left( e^{(A-1)B|t|} - 1 \right) e^{B|t|}, & A \neq 1, \\ B|t| e^{B|t|}, & A=1. \end{cases}
\ee
This proves Lemma~\ref{lem2}. One may choose
\be \label{eq:consts}
C = \frac{C_{\mu}}{2 \mu v} \left( |y_1| \|v_+\|_{\infty} + |y_2| \|v_-\|_{\infty} \right)
\ee
and by setting $\beta = \frac{\| L(0) \|_2 \lambda_+}{2 \mu v} - 1$, we see that
\be \label{eq:defh}
h(t) = \begin{cases} \frac{1}{\beta} \left(e^{2 \mu v \beta |t|} - 1 \right) & \mbox{if } \beta \neq 0, \\ 2 \mu v |t| & \mbox{if } \beta =0 \, . \end{cases}
\ee 
Note that $\beta$ depends only on $\mu$, $\lim_{t \to 0} h(t) =0$, and moreover, if $\beta <0$, then $h(t) \leq | \beta |^{-1} $.
For any $\mu$, however, $h(t) \leq h e^{r|t|}$, where $h = h( \mu, \x)$ and $r = r(\mu, \x) \geq 0$.
\end{proof}

Our proof of Theorem~\ref{thm:pert_toda_interpol} will again use an iteration scheme. Due to interpolation, it will
be necessary to sum certain terms over all integers. To this end, we make the following observation. 
For each $\mu >0$, the function
\be  \label{eq:defgmu}
G_\mu(k) = \frac{e^{- \mu|k|}}{(1+|k|)^2}
\ee 
satisfies the estimate
\be \label{eq:itbd}
\sum_{l \in \mathbb{Z}}G_\mu(j-l)G_\mu(l-k) \leq \gamma G_\mu(j-k)
\ee
with $\gamma = 4 \sum_{k \in \mathbb{Z}}(1+|k|)^{-2}$. In fact, setting $f(k)=(1+|k|)^2$, it is clear that
$f(j+k) \le 2(f(j)+f(k))$, and thus
\bea
G_\mu(j-k)^{-1}  G_\mu(j-l)G_\mu(l-k) & = & f(j-k) f(j-l)^{-1} f(l-k)^{-1} e^{\mu(|j-k|-|j-l|-|l-k|)} \nonumber \\
& \le & 2(f(j-l)^{-1} + f(l-k)^{-1}) \, .
\eea
The claim follows after summing over $l$. Clearly the choice of power $2$ in the denominator
of $G_\mu$ is merely for convenience; a factor $(1+|k|)^{-1 - \delta}$ for any $\delta >0$ would suffice.
Note, however, that there is no bound of the type in \eqref{eq:itbd} for the exponential function with no inverse
polynomial weight.  

As a final comment, it is clear that for any $\mu >0$ and $\epsilon >0$, the bound
\be \label{eq:ebdg}
e^{-(\mu + \epsilon)|x|} \leq C_{\epsilon} G_{\mu}(|x|) 
\ee
holds for all $x \in \mathbb{R}$. The number $C_{\epsilon} = \sup_{x \in \mathbb{R}}(1+|x|)^2 e^{- \epsilon |x|} < \infty$.

We can now state the second result of this section.

\begin{thm}\label{thm:pert_toda_interpol}
Fix $W \in C^2(\mathbb{R})$ with $W', W'' \in L^{\infty}(\mathbb{R})$, take $\mu>0$, and let $\x \in M_b$.
For any $\epsilon >0$, there are positive numbers $C = C( \epsilon)$, $D = D( \epsilon, \mu, \x, W)$, and $\delta =\delta(\epsilon, \mu, \x, W)$ such that
\be \label{eq:interbd}
\max \left[ \abs{ \fr{\D}{ \D z} a_n^\w(t)}, \abs{ \fr{\D}{ \D z} b_n^\w(t)} \right] \leq C G_{\mu}( |n-m| ) e^{(\mu+\epsilon) v |t|} \left[ 1 + D \left( e^{\delta |t|} -1 \right) \right] 
\ee
holds for all $t\in\R$ and $n, m \in \mathbb{Z}$. Here $v = v(\x, \mu + \epsilon)$ is as in Theorem \ref{thm:toda_solest} and $z\in\{a_m,b_m\}$.
\end{thm}
In words, this result shows that for each $\mu >0$ the perturbed Lieb--Robinson velocity $v^{\w}(\x, \mu)$, corresponding to $\x \in M_b$, satisfies, for
each $\epsilon >0$,
\be
v^{\w}(\x, \mu) \leq \left( 1 + \frac{\epsilon}{\mu} \right)v(\x, \mu+ \epsilon) + \frac{\delta}{\mu}
\ee
and, as is shown in the proof below, the dependence of $\delta$ on $W$ can be made explicit
\be
\delta = \delta_1 \| W' \|_{\infty} + \delta_2 \| W'' \|_{\infty} + \delta_3
\ee
for some $\delta_i = \delta_i(\epsilon, \mu, \x)$ and $i = 1,2,3$.
\begin{proof}
Fix $W$, $\x \in M_b$, and $t \in \mathbb{R}$. We begin by interpolating between the Toda and perturbed dynamics. 
As in Section~\ref{sec:toda}, denote by 
\be
\alpha_t(A) = A \circ \Phi_t \quad \mbox{and} \quad \alpha^{\w}_t(A) = A \circ \Phi_t^{\w} \quad \mbox{for any } A \in \mathcal{A} ,
\ee
the Toda and perturbed dynamics respectively. For any $A \in \mathcal{A}$, the observable-valued equation
\be\label{eq:int}
\al_t^\w(A)-\al_t(A)=\int_0^t\fr{d}{ds}\al_s^\w(\al_{t-s}(A))ds
\ee
is clear. Due to the form of the perturbation,
\bea \label{eq:der}
\fr{d}{ds}\al_s^\w(\al_{t-s}(A))&=&\al_s^\w(\{\al_{t-s}(A),H^\w\})-\al_s^\w(\al_{t-s}(\{A,H \})) \nonumber \\
&=&\sum_{k\in\Z}\al_s^\w(\{\al_{t-s}(A),W_k\}) \, ,
\eea
and moreover, one finds that 
\be \label{eq:calc}
\{\al_{t-s}(A),W_k\} = - \fr{1}{2} W_k' \cdot \fr{\D}{ \D \tilde{b}_k} \alpha_{t-s}(A) \, .
\ee
Combining these expressions, we have shown that
\be
\al_t^\w(A) = \al_t(A) - \fr{1}{2} \sum_{k \in \mathbb{Z}} \int_0^t \alpha_s^{\w} \left( W_k' \cdot \fr{\D}{ \D \tilde{b}_k} \alpha_{t-s}(A) \right) \, ds 
\ee
as a formal expansion. For specific choices of observables, the expression above can be estimated. 

Fix $n \in \mathbb{Z}$ and denote by $A_n$ and $B_n$ the observables introduced in Example~\ref{ex:1}.
Consider the observable-valued 
\be
\tilde{F}^{\w}_n(t) = \begin{pmatrix} \alpha_t^\w(A_n) \\ \alpha_t^\w(B_n) \end{pmatrix} \ \ \mbox{and}\ \ \tilde{F}_n(t) = \begin{pmatrix} \alpha_t(A_n) \\ \alpha_t(B_n) \end{pmatrix}, 
\ee
which for any given $\x \in M$ satisfies
\be
[\tilde{F}^{\w}_n(t)](\x) =   \begin{pmatrix} a_n^{\w}(t) \\ b_n^{\w}(t) \end{pmatrix} = F^{\w}_n(t)\ \ \mbox{and}\ \ [\tilde{F}_n(t)](\x) =   \begin{pmatrix} a_n(t) \\ b_n(t) \end{pmatrix}=F_n(t),
\ee
in terms of our previous notation. Our starting point is the equation
\be
\tilde{F}^{\w}_n(t) = \tilde{F}_n(t) - \fr{1}{2} \sum_{k \in \mathbb{Z}} 
\int_0^t  \alpha_s^\w(W_k') \cdot \begin{pmatrix} \alpha_s^\w \left( \fr{\D}{\D \tilde{b}_k} \alpha_{t-s}(A_n) \right) \\ 
\alpha_s^\w \left( \fr{\D}{\D \tilde{b}_k} \alpha_{t-s}(B_n) \right) \end{pmatrix} \, ds \, .
\ee
Since $W$ is sufficiently smooth, the formula 
\be \label{eq:intdz}
\fr{\D}{\D z} \tilde{F}^{\w}_n(t) = \fr{\D}{\D z} \tilde{F}_n(t) -  \sum_{k \in \mathbb{Z}} 
\int_0^t D_k(s;n) \fr{\D}{\D z} \tilde{F}^{\w}_k(s) \, ds - \fr{1}{2} \sum_{k, \ell \in \mathbb{Z}} 
\int_0^t \alpha_s^{\w}(W_k') \cdot D_{k, \ell}(s;n) \fr{\D}{\D z} \tilde{F}^{\w}_{\ell}(s) \, ds 
\ee
where
\be \label{eq:1dermat}
D_k(s;n) = \begin{pmatrix} \alpha_s^{\w} \left( W_k'' \cdot A_k^{-1} \right) \cdot \alpha_s^{\w} \left( \fr{\D}{\D \tilde{b}_k} \alpha_{t-s}(A_n) \right) & 0 
\\  \alpha_s^{\w} \left( W_k'' \cdot A_k^{-1} \right) \cdot \alpha_s^{\w} \left( \fr{\D}{\D \tilde{b}_k} \alpha_{t-s}(B_n) \right) & 0 \end{pmatrix}
\ee
and
\be \label{eq:2dermat}
D_{k, \ell}(s;n) = \begin{pmatrix}  \alpha_s^\w \left( \fr{\D^2}{ \D a_{\ell} \D \tilde{b}_k} \alpha_{t-s}(A_n) \right) & \alpha_s^\w \left( \fr{\D^2}{ \D b_{\ell} \D \tilde{b}_k} \alpha_{t-s}(A_n) \right) \\
 \alpha_s^\w \left( \fr{\D^2}{ \D a_{\ell} \D \tilde{b}_k} \alpha_{t-s}(B_n) \right) & \alpha_s^\w \left( \fr{\D^2}{ \D b_{\ell} \D \tilde{b}_k} \alpha_{t-s}(B_n) \right) \end{pmatrix}
\ee
readily follows.

For any $\x \in M$, the estimate
\be \label{eq:1intbd}
\begin{split}
\left\| \left[ \fr{\D}{\D z} \tilde{F}^{\w}_n(t) \right] (\x) \right\|_{\infty}  & \leq  \left\| \left[ \fr{\D}{\D z} \tilde{F}_n(t)  \right] (\x) \right\|_{\infty}  +  \sum_{k \in \mathbb{Z}} 
\int_0^{|t|} \left\| \left[  D_k(s;n) \fr{\D}{\D z} \tilde{F}^{\w}_k(s) \right] (\x) \right\|_{\infty} \, ds \\
& \quad \quad + \fr{\| W'\|_{\infty} }{2} \sum_{k, \ell \in \mathbb{Z}} 
\int_0^{|t|}  \left\| \left[ D_{k, \ell}(s;n) \fr{\D}{\D z} \tilde{F}^{\w}_{\ell}(s) \right] (\x) \right\|_{\infty} \, ds
\end{split}
\ee
is clear. 

Using Theorem~\ref{thm:toda_solest}, the first term can be bounded by 
\bea
\left\| \left[ \fr{\D}{\D z} \tilde{F}_n(t)  \right] (\x) \right\|_{\infty} & = & \left\| \fr{\D}{\D z} F_n(t)  \right\|_{\infty} \\
& \leq & \frac{8}{\sqrt{17}}e^{-(\mu + \epsilon)\left( |n-m| - v|t| \right)} \nonumber \\
& \leq & \frac{8}{\sqrt{17}} C_{\epsilon} G_{\mu}(|n-m|) e^{(\mu + \epsilon) v|t|} \nonumber
\eea
where we have set $v = v(\x, \mu+ \epsilon)$ and used \eqref{eq:ebdg}. 

For each $\x \in M_b$, the matrix appearing in the second term satisfies
\bea
\left\| \left[  D_k(s;n) \right] (\x) \right\|_{\infty} & \leq & C_2 \| W'' \|_{\infty} \frac{8}{\sqrt{17}} e^{\mu+ \epsilon} e^{- (\mu+\epsilon) \left(|n-k| -v(s)(|t|-s) \right)} \\
& \leq & C_2 \| W'' \|_{\infty} \frac{8}{\sqrt{17}} e^{\mu+ \epsilon} C_{\epsilon} G_{\mu}(|n-k|) e^{(\mu+\epsilon)v(s)(|t|-s)} \nonumber
\eea
where we have denoted by
\be \label{eq:lrbs}
v(s) = v(\Phi_s^{\w}(\x), \mu + \epsilon) = (1+\sqrt{17}) \left\|L(\Phi_s^{\w}(\x)) \right\|_2 \left(e^{\mu + \epsilon+1} + (\mu+\epsilon)^{-1}\right)
\ee
the unperturbed Lieb--Robinson corresponding to the initial condition $\Phi_s^{\w}(\x)$. Since $\x \in M_b$,
\be
\left\|L(\Phi_s^{\w}(\x)) \right\|_2 \leq 2 \left\| a^\w(s) \right\|_\infty+\left\|b^\w(s)\right\|_\infty \leq 3 C_1 \, ,
\ee
and so the quantity in \eqref{eq:lrbs} can be estimated independent of $s$. Clearly,
\be
v = v(0) \leq \sup_{s \in \mathbb{R}}v(s) : = v^*.
\ee
 
The matrix in the third term can be dominated using Lemma~\ref{lem2}. In fact, the bound
\be\label{secondmatrix}
\left\| \left[  D_{k, \ell}(s;n) \right] (\x) \right\|_{\infty} \leq 2 C e^{- (\mu + \epsilon)|n- \ell|} e^{- (\mu+\epsilon) |n-k|} e^{2 (\mu +\epsilon) v(s)(|t|-s)} h(|t|-s) 
\ee
follows immediately. Some comments are in order. First, the prefactor $C$ from Lemma~\ref{lem2}, appearing above, seems to depend on $s$ through the velocity $v(s)$; see \eqref{eq:consts}. 
If, however, one repeats the argument of Lemma~\ref{lem2} and replaces the $v(s)$ in \eqref{eq:intvel} with $v^*$ above, then the
new prefactor is independent of $s$. Next, as discussed at the end of Lemma~\ref{lem2}, the function $h$ grows at most exponentially.
In fact, in this particular application, it is clear that
\be
h(|t|-s) \leq h e^{2( \mu + \epsilon) \gamma v^*(|t|-s)}
\ee
with numbers $h = h(\x, \mu + \epsilon)$ and $\gamma = \gamma(\x, \mu + \epsilon)$, each independent of $s$. This proves that
\be \label{eq:thirdmatrix}
\left\| \left[  D_{k, \ell}(s;n) \right] (\x) \right\|_{\infty} \leq  \tilde{C} G_{\mu}( |n- \ell|) G_{\mu}( |n-k|) e^{2(\gamma +1)(\mu +\epsilon) v^* (|t|-s)} \, .
\ee 

Putting everything together and suppressing the $\x$-dependence, we have found that
\bea 
\left\| \fr{\D}{\D z} F^{\w}_n(t) \right\|_{\infty}  & \leq & B_1 G_{\mu}(|n-m|) e^{(\mu+\epsilon)v|t|} +  B_2 \sum_{k \in \mathbb{Z}} G_{ \mu}( |n-k|)
\int_0^{|t|} e^{(\mu+ \epsilon) v^* (|t|-s)} \left\|  \fr{\D}{\D z}F^{\w}_k(s) \right\|_{\infty} \, ds \nonumber \\
& \mbox{ } &   \quad + B_3 \sum_{ k, \ell \in \mathbb{Z}} G_{\mu}( |n- \ell|)G_{\mu}(|n-k|)  \int_0^{|t|} e^{2 (\gamma+1)(\mu+\epsilon) v^* (|t|-s)} \left\| \fr{\D}{\D z} F^{\w}_{\ell}(s)\right\|_{\infty} \, ds \nonumber \\
& \leq & B_1 G_{\mu}(|n-m|) e^{(\mu+\epsilon)v|t|} +  B \sum_{k \in \mathbb{Z}} G_{ \mu}( |n-k|)
\int_0^{|t|} e^{2(\gamma+1)(\mu+ \epsilon) v^* (|t|-s)} \left\|  \fr{\D}{\D z}F^{\w}_k(s) \right\|_{\infty} \, ds
\eea
Upon iteration, we find that
\be \nonumber
\begin{split}
\left\|\fr{\D}{\D z}F^{\w}_n(t) \right\|_{\infty}  & \leq B_1 G_{\mu}( |n-m| ) e^{a|t|}  + 
B_1 \sum_{j=1}^{\infty} B^j \sum_{k_1 \in \mathbb{Z}} \cdots \sum_{k_j \in \mathbb{Z}} G_{\mu}(|n-k_1|) \cdots G_{\mu}(|k_j-m|) \\
& \quad \quad \times \int_0^{|t|} \int_0^{s_1} \cdots \int_0^{s_{j-1}} e^{b(|t|-s_1)} e^{b(s_1-s_2)} \cdots e^{b(s_{j-1}-s_j)} e^{ a s_j} ds_j \cdots d s_1 
\end{split}
\ee
where we have set
\be
a = (\mu + \epsilon)v \quad \mbox{and} \quad b = 2(\gamma +1) (\mu+\epsilon)v^* \, . 
\ee
{F}rom \eqref{eq:itbd}, it is clear that
\be
\sum_{k_1 \in \mathbb{Z}} \cdots \sum_{k_j \in \mathbb{Z}} G_{\mu}(|n-k_1|) \cdots G_{\mu}(|k_j-m|)  \leq \gamma^j G_{\mu}(|n-m|)
\ee
and the iterated integral can also be calculated:
\bea
\int_0^{|t|} \int_0^{s_1} \cdots \int_0^{s_{j-1} } e^{b(|t|-s_1)} e^{b(s_1-s_2)} \cdots e^{b(s_{j-1}-s_j)} e^{ a s_j} ds_j \cdots d s_1 
&=& e^{b|t|} \int_0^{|t|} \int_0^{s_1} \cdots \int_0^{s_{j-1}} e^{(a-b) s_j} ds_j \cdots d s_1 \nonumber \\
&=&  \fr{e^{b|t|}}{(a-b)^j} \sum_{k=j}^{\infty} \fr{ \left( (a -b) |t| \right)^k}{k!}
\eea
This shows that
\be
\left\| \fr{\D}{\D z} F^{\w}_n(t) \right\|_{\infty}   \leq B_1 G_{\mu}( |n-m| ) \left( e^{a|t|} + 
e^{b|t|} \sum_{j=1}^{\infty} \left( \fr{B \gamma}{a-b} \right)^j \sum_{k=j}^{\infty} \fr{ \left((a-b)|t| \right)^k}{k!} \right)
\ee
Since $v_{\mu + \epsilon} \leq v^*$, it is clear that $a-b<0$. In this case, 
\be
\sum_{j=1}^{\infty} \left( \fr{B \gamma}{a-b} \right)^j \sum_{k=j}^{\infty} \fr{ \left((a-b)|t| \right)^k}{k!} = \frac{B \gamma}{B\gamma + b -a} \left[ e^{B \gamma |t|} - e^{(a-b)|t|} \right]
\ee
follows from \eqref{eq:bigsum} and so 
\be
\left\| \fr{\D}{\D z} F^{\w}_n(t) \right\|_{\infty}   \leq B_1 G_{\mu}( |n-m| ) e^{a|t|} \left( 1 + \frac{B \gamma}{B\gamma + b -a}
\left[ e^{ \left(B \gamma +b - a \right) |t|} - 1 \right] \right) \, .
\ee
This proves \eqref{eq:interbd} and completes the proof.
\end{proof}

\subsection{Comments on these results}
It is clear that results analogous to Theorem~\ref{thm:pert_toda_solest} and Theorem~\ref{thm:pert_toda_interpol} hold for more general, finite range potentials $W$.
We stated the results as above for simplicity of presentation.

It is also clear that analogues of Theorem~\ref{thm:toda_lrb}, for more general observables, hold as direct corollaries of Theorem~\ref{thm:pert_toda_solest} and Theorem~\ref{thm:pert_toda_interpol}. 
Since the statements are clear, we do not rewrite them for the sake of brevity.

%
%
%
%
%

\section{Bounds for the Hierarchy}\label{sec:hth}

In this section, we will demonstrate that the results from Section~\ref{sec:loctoda} also apply to
the Toda hierarchy. We begin, in Section~\ref{subsec:todahier}, by introducing the Toda hierarchy using the 
recursive approach from \cite{Bulla}. We will use essentially the same notation as \cite{Te} and
refer the interested reader to this text for further information. In Section~\ref{subsec:aprisolest}, we 
state the Lieb--Robinson bound valid for the hierarchy, see Theorem~\ref{thm:hier_lrb}. To prove this result, we 
first establish the crucial solution estimate, an analogue of Theorem~\ref{thm:toda_solest}, which shows that a solution of the
hierarchy at site $n$ has weak dependence on the initial condition at site $m$ if $|n-m| >>1$. The bound is
explicit in terms of a quantity $v_r$ which, in particular, depends on the initial condition through the infinity norm of
a specific matrix, see \eqref{eq:velhier}. To make our bound even more concrete, we estimate the norm of this matrix in Lemma~\ref{lem:hiervel} below. 
Theorem~\ref{thm:hier_lrb} follows as in Section~\ref{sec:loctoda}.

We note that our results also immediately apply to the Kac--van Moerbeke hierarchy.
This hierarchy can be viewed as a special case of the Toda hierarchy which is 
obtained by setting $b_n\equiv 0$ in the even order Toda equations.  As is discussed, e.g.\ in \cite{mite2009}, 
this gives precisely the equations of the Kac--van Moerbeke hierarchy.

\subsection{The Toda Hierarchy} \label{subsec:todahier}
We introduce the Toda Hierarchy as follows. Fix $r\in \N_0 = \N \cup \{0\}$. Take $c_0=1$ and choose constants $c_j \in \mathbb{R}$ for $1\le j \le r$.
For each $\x \in M$, define sequences $g(\x, r)$ and $h(\x, r)$ componentwise by setting
\be \label{eq:gj}
g_n(\x, r) = \sum_{j=0}^{r} c_{r-j} \ti{g}^{(j+1)}_n(\x) \quad \mbox{with} \quad \ti{g}^{(j)}_n(\x)=\spr{\delta_n}{L(\x)^j \delta_n}, 
\ee
and
\be \label{eq:hj}
h_n(\x, r) = \sum_{j=0}^{r} c_{r-j}  \ti{h}^{(j+1)}_n(\x)  \quad \mbox{with} \quad \ti{h}^{(j)}_n(\x) =  2 a_n \spr{\delta_{n+1}}{L(\x)^j\delta_n}, 
\ee
with $L(\x)$ as in \eqref{eq:defL}. We now define a system of equations for
the components of an unknown sequence $\x(t,r) = \{ (a_n(t,r), b_n(t,r) \}$: 
\be \label{eq:HTH_1}
\dot{a}_n(t, r) = a_n(t, r)\left(g_{n+1}( t, r)-g_n( t,r)\right) 
\ee
and
\be \label{eq:HTH_2}
\dot{b}_n(t, r) = \left(h_n( t, r)-h_{n-1}(t, r)\right)  \, 
\ee
with initial condition $\x(0) = \{ (a_n(0,r), b_n(0,r) \} = \x \in M$. 
As in the previous section, we have, for example, denoted by $g_n(t,r) = g_n(\x(t,r), r)$ to simplify notation.
The system \eqref{eq:HTH_1} and \eqref{eq:HTH_2} is known to have global solutions (see e.g.\ \cite{Te}, Theorem 12.6)
for initial conditions in $M$.
Varying $r\in\N_0$ describes the Toda hierarchy. Let us denote by $\Phi_t^{(r)}$ the flow corresponding
to the Toda hierarchy, i.e.\ $\Phi_t^{(r)}: M \to M$ satisfies $\Phi_t^{(r)}(\x) = \{ (a_n(t,r), b_n(t,r) \}$, the solution of
\eqref{eq:HTH_1} and \eqref{eq:HTH_2} above, with $\Phi_0^{(r)}(\x) = \x$.

The simplest example corresponds to $r=0$. In this case, the system becomes
\bea
\dot{a}_n(t,0) = a_n(t,0)(b_{n+1}(t,0)-b_n(t,0)) \quad \mbox{and} \quad \dot{b}_n(t,0) = 2(a_n(t,0)^2-a_{n-1}(t,0)^2) \, ,
\eea
which is, of course, the Toda system \eqref{eq:toda}.

\subsection{The Lieb--Robinson Bound} \label{subsec:aprisolest}
The above choice of sequences $g$ and $h$ guarantee that the Lax-formalism of the
Toda Lattice, see \eqref{eq:defL}--\eqref{eq:Lax}, still holds for the hierarchy. In fact, set
\be
P(\x, r) = \sum_{j=0}^r c_{r-j} \ti{P}^{(j+1)}(\x) \quad \mbox{with} \quad
\ti{P}^{(j)}(\x) = [L(\x)^{j}]_+ - [L(\x)^{j}]_-,
\ee
where $[A]_\pm$ denote the upper and lower triangular parts of an operator with respect to the
standard basis $\delta_m(n)= \delta_{m,n}$ (with $\delta_{m,n}$ the usual Kronecker delta).
It is known, see e.g.\ \cite{Te}, that the Toda hierarchy is equivalent to the Lax equation
\be
\fr{d}{dt}L(t)=[P(t, r),L(t)] \, ,
\ee
where we have set $L(t) = L(\Phi_t^{(r)}(\x))$ and $P(t,r) = P(\Phi_t^{(r)}(\x), r)$ again to ease notation.
It is easy to see that the operator $P(t,r)$, which is of order $2r+2$, 
is skew-adjoint and differentiable. Like before then, there exists a unique unitary propagator 
$U^{(r)}(t,s)$ for $P(t,r)$. It follows from the Lax equation that
\be
L(t)=U^{(r)}(t, s)L(s)U^{(r)}(t, s)^{-1}
\ee
implying again an a-priori estimate
\be\label{hth_unifest}
\max\left(\left\|a(t,r)\right\|_\infty,\left\|b(t,r)\right\|_\infty\right)\leq \left\|L(t)\right\|_2=\left\|L(0)\right\|_2 \, .
\ee

For later use we record the following structure. 

\begin{lem}[\cite{KruegerTeschl2}]\label{lem:structgh}
For each $\x \in M$ and any integer $j \geq 1$, the sequences $\ti{g}^{(j)}(\x)$ and $\ti{h}^{(j)}(\x)$, as defined in \eqref{eq:gj} and \eqref{eq:hj}, 
have components that are homogeneous. In fact, they are sums of monomials of the components of $\x$
with degree $j$ and $j+1$, respectively, which have the form:
\[
\ti{g}^{(j)}_n(\x) = \begin{cases}
\left(\prod\limits_{\ell=0}^{k-1} a_{n+\ell}^2 \right) b_{n+k} + R(n+k-1,n-k+1) +\\
+ \left(\prod\limits_{\ell=1}^k a_{n-\ell}^2 \right) \left( b_{n-k} + 2 \sum\limits_{\ell=0}^{k-1} b_{n-\ell}\right), \qquad j=2k+1,\\
 \left(\prod\limits_{\ell=0}^{k-2} a_{n+\ell}^2 \right) \Big( a_{n+k-1}^2 +  b_{n+k-1}^2 + 2b_{n+k-1}\sum\limits_{\ell=0}^{k-2} b_{n+\ell} \Big)+ \\
+ R(n+k-2,n-k+1) + \prod\limits_{\ell=1}^k a_{n-\ell}^2, \qquad j=2k,
\end{cases}
\]
and
\[
\ti{h}^{(j)}_n(\x) = \begin{cases}
2 \left(\prod\limits_{\ell=0}^{k-1} a_{n+\ell}^2 \right) \Big( a_{n+k}^2 +  b_{n+k}^2 + 2b_{n+k}\sum\limits_{\ell=0}^{k-1} b_{n+\ell} \Big)+\\
+ R(n+k-1,n-k+1) + 2\prod\limits_{\ell=0}^k a_{n-\ell}^2, \qquad j=2k+1,\\
2\left(\prod\limits_{\ell=0}^{k-1} a_{n+\ell}^2 \right) b_{n+k} + R(n+k-1,n-k+2) +\\
+ 2 \prod\limits_{\ell=0}^{k-1} a_{n-\ell}^2 \left( b_{n+1} + b_{n-k+1} +2 \sum\limits_{\ell=0}^{k-2} b_{n-\ell} \right), \qquad j=2k,
\end{cases}
\]
for $j>1$. Here $R(n,m)$ denotes terms which involve only $a_\ell$ and $b_\ell$ with $m \le \ell \le n$ and
we set $R(n,m)=0$ if $n<m$.
\end{lem}

The result below is an analogue of Theorem \ref{thm:toda_solest} for the Toda hierarchy. 
\begin{thm} \label{thm:hth_solest} Fix $\x \in M$, $\mu >0$ and $r\in\N_0$. For any $n,m \in \Z$, the bound 
\be  \label{eq:lrbhier}
\max \left[ \abs{ \fr{\D}{ \D z} a_n(t,r)}, \abs{ \fr{\D}{ \D z} b_n(t,r)} \right] \leq 
e^{- \mu \left( \lceil |n-m|/ (\lfloor\frac{r}{2}\rfloor+1)\rceil - v_r |t| \right)} \, ,
\ee
holds for all $t\in\R$, where $z\in\{a_m,b_m\}$,
\be \label{eq:velhier}
v_r = v_r(\x,\mu)= \| D(r) \|_{\infty} \| L(0) \|_2 \left(e^{\mu+1}+\fr{1}{\mu}\right) \, ,
\ee
and $D(r)$ depends on $\x$, $r$, and the numbers $c_1,\dots,c_r$.
\end{thm}
\begin{proof}
Without loss of generality assume $t\geq 0$. For each $n\in\Z$, define the function $F_n:\R\rightarrow \R^2$ by
\be
F_n(t,r)=\begin{pmatrix} a_n(t,r) \\ b_n(t,r) \end{pmatrix}.
\ee
Since $r$ will be fixed for the remainder of the argument, we will drop it from our notation.
Using the equations of motion \eqref{eq:HTH_1} and \eqref{eq:HTH_2}, it is clear that
\bea 
F_n(t) & = & F_n(0)+\int_0^t \begin{pmatrix}  a_n(s) \left( g_{n+1}(s) - g_n(s) \right) \\  h_n(s) -h_{n-1}(s)  \end{pmatrix}ds \\
& = & F_n(0)+\sum_{j=0}^{r} c_{r-j} \int_0^t \begin{pmatrix}  a_n(s) \left( \ti{g}_{n+1}^{(j+1)}(s) - \ti{g}_n^{(j+1)}(s) \right) \\  \ti{h}_n^{(j+1)}(s) -\ti{h}_{n-1}^{(j+1)}(s)  \end{pmatrix}ds \nonumber
\eea
Observe that by Lemma~\ref{lem:structgh}, for each $0 \leq j \leq r$, the quantities 
$a_n(s)(\ti{g}_{n+1}^{(j+1)}(s) - \ti{g}_n^{(j+1)}(s))$ and $\ti{h}_n^{(j+1)}(s) -\ti{h}_{n-1}^{(j+1)}(s)$ 
are homogeneous polynomials (of degree $j+2$) in the variables $a_{n+e}(s)$ and $b_{n+e}(s)$ for $|e|\leq \lfloor\frac{j}{2}\rfloor+1$. 

Differentiating with respect to $z\in\{a_m,b_m\}$ we get
\be 
\fr{\D}{ \D z} F_n(t) = \fr{\D}{ \D z} F_n(0) +\sum_{j=0}^{r} c_{r-j} \sum_{|e|\leq \lfloor\frac{j}{2}\rfloor+1}\int_0^t D_{n,e}^{(j)}(s) \fr{\D}{ \D z} F_{n+e}(s)ds,
\ee
where the entries  of $D_{n,e}^{(j)}(s)$ are homogeneous polynomials of degree $j+1$. In particular, the following estimate, analogous to the bound in \eqref{eq:1bd}, 
holds
\be \label{eq:1bdonD}
|D_{n,e}^{(j)}(s)|\leq  \| L(0) \|_2^{j+1} \, D_{e}^{(j)} \, ,
\ee
an explicit formula for $D_e^{(j)}$ appears in \eqref{eq:Dej} below. As a result, it is clear that
\be
\left|\fr{\D}{ \D z} F_n(t)\right| \leq \left|\fr{\D}{ \D z} F_n(0)\right| + \sum_{j=0}^{r} |c_{r-j}| \| L(0) \|_2^{j+1} \sum_{|e|\leq \lfloor\frac{j}{2}\rfloor+1}\int_0^t D_{e}^{(j)} \left|\fr{\D}{ \D z} F_{n+e}(s)\right| ds \, .
\ee
Upon iteration of the above inequality (in the case that $z = a_m$), we find that
\bea
\left|\fr{\D}{ \D a_m} F_n(t)\right| &\leq& \sum_{k=0}^\infty  \fr{t^k}{k!} \sum_{j_1 =0}^{r} |c_{r-j_1}| \| L(0) \|_2^{j_1+1} \cdots \sum_{j_k =0}^{r} |c_{r-j_k}| \| L(0) \|_2^{j_k+1} \times \nonumber \\
& \mbox{ } & \quad \times \sum_{|e_1|\leq \lfloor\frac{j_1}{2}\rfloor+1} \cdots \sum_{|e_k| \leq \lfloor\frac{j_k}{2}\rfloor+1} 
\delta_{m+e_1+\cdots + e_k}(n) D_{e_1}^{(j_1)} \cdots D_{e_k}^{(j_k)} \begin{pmatrix} 1 \\ 0 \end{pmatrix} \nonumber \\ \nonumber
&\le&\sum_{k=\lceil |n-m|/ (\lfloor\frac{r}{2}\rfloor+1)\rceil}^\infty \fr{(\| L(0) \|_2 t)^k}{k!} \sum_{j_1 =0}^{r} |c_{r-j_1}| \| L(0) \|_2^{j_1} \cdots \sum_{j_k =0}^{r} |c_{r-j_k}| \| L(0) \|_2^{j_k} 
D^{(j_1)} \cdots D^{(j_k)} \begin{pmatrix} 1 \\ 0 \end{pmatrix} \, \nonumber \\
& \le&\sum_{k=\lceil |n-m|/ (\lfloor\frac{r}{2}\rfloor+1)\rceil}^\infty \fr{(\| L(0) \|_2 t)^k}{k!} D(r)^k \begin{pmatrix} 1 \\ 0 \end{pmatrix} \, 
\eea
where we have set 
\be \label{eq:defD}
D^{(j)} = \sum_{|e|\leq \lfloor\frac{j}{2}\rfloor+1} D_{e}^{(j)}  \quad \mbox{and} \quad D(r) = \sum_{j=0}^{r} |c_{r-j}| \| L(0) \|_2^j D^{(j)} \, .
\ee
A similar estimate holds for the case $z=b_m$.

Taking the infinity norm one obtains that for any $\mu>0$
\bea
\left\| \fr{\D}{\D z}F_n(t) \right\|_{\infty} &\leq& \sum_{k=\lceil |n-m|/ (\lfloor\frac{r}{2}\rfloor+1)\rceil}^\infty \fr{ (\| L(0) \|_2 t)^k}{k!} \|D(r)\|_\infty^k\\
&\leq& e^{-\mu\left(\lceil |n-m|/ (\lfloor\frac{r}{2}\rfloor+1)\rceil-v_r|t|\right)}, \nonumber
\eea
where we have set $v_r = \|D(r)\|_{\infty} \| L(0) \|_2 \left(e^{\mu+1}+\fr{1}{\mu}\right)$. 
\end{proof}

We now provide a rough, but explicit, estimate on the velocity corresponding to the Toda hierarchy.
\begin{lem} \label{lem:hiervel} Fix $\x \in M$, $\mu >0$ and $r\in\N_0$. The velocity corresponding to the Toda hierarchy, see \eqref{eq:lrbhier} and
\eqref{eq:velhier} satisfies
\be
v_r \leq 8 \left( e^{\mu+1} + \frac{1}{\mu} \right) \sum_{j=0}^{r} |c_{r-j}| \| L(0) \|_2^{j+1} (j+2) 3^j
\ee
\end{lem}
\begin{proof}
It is clear from \eqref{eq:velhier} that we need only estimate the quantity $\| D(r) \|_{\infty}$, with $D(r)$ as defined in \eqref{eq:defD}.
To see this, first recall that for each $1 \leq j \leq r$, the quantities $a_n(s)(\ti{g}_{n+1}^{(j+1)}(s) - \ti{g}_n^{(j+1)}(s))$ and $\ti{h}_n^{(j+1)}(s) -\ti{h}_{n-1}^{(j+1)}(s)$ 
are homogeneous polynomials (of degree $j+2$) in the variables $a_{n+e}(s)$ and $b_{n+e}(s)$ for $|e|\leq \lfloor\frac{j}{2}\rfloor+1$.
We need an estimate on the number of terms in each of these polynomials. 

Both sequences $\ti{g}^{(j+1)}$ and $\ti{h}^{(j+1)}$, see \eqref{eq:gj} and \eqref{eq:hj}, have components defined in terms of 
the operator $L(\x)$ given by \eqref{eq:defL}.
In terms of the shifts $S^{\pm}$ on $\ell^2(\mathbb{Z})$, i.e.\ $(S^{\pm}f)_n = f_{n \pm 1}$, one can write 
\be
L(\x) = aS^+ +a^-S^- + b \, .
\ee
Here $a$ and $b$ are regarded a multiplication operators and $a^{\pm}$ is the multiplication operator given by
$(a^{\pm}f)_n = a_{n \pm 1}f_n$. An upper bound on the desired number of terms can be obtained by taking
both $a$ and $b$ to be constant sequences, i.e., calculating
\be
\eta^{(j+1)} = \left\langle \delta_n, \left(S^+ + S^- +I \right)^{j+1} \delta_n \right\rangle = \sum_{k=0,\ k\ even}^{j+1} {j+1\choose k} {k\choose k/2}
\ee
and
\be
\xi^{(j+1)} = \left\langle \delta_{n+1},\left( S^++S^-+I\right)^{j+1} \delta_n \right\rangle = \sum_{k=0,\ k\ odd}^{j+1} {j+1\choose k} {k\choose (k+1)/2}
\ee
The bounds $\eta^{(j+1)} \leq 2 \xi^{(j+1)}$ and $\xi^{(j+1)} \leq 3^j$ readily follow. 

Next, using this estimate, we can expand
\be \label{eq:expgj}
a_n(s) \left( \ti{g}_{n+1}^{(j+1)}(s) - \ti{g}_n^{(j+1)}(s) \right) = \sum_{k=1}^{2\eta^{(j+1)}} \prod_{\ell =1}^{j+2} d_{n+e_{k,\ell}}
\ee
and
\be \label{eq:exphj}
\ti{h}_n^{(j+1)}(s) -\ti{h}_{n-1}^{(j+1)}(s) = 2 \sum_{k=1}^{2\xi^{(j+1)}} \prod_{\ell =1}^{j+2} d_{n+e'_{k,\ell}}
\ee
where $d_{n+e} \in \{ \pm a_{n+e}(s), \pm b_{n+e}(s) \}$. Here we have inserted a $2$ because of the definition of 
$\ti{h}^{(j+1)}$, see \eqref{eq:hj}. This form enables us to determine the entries of
the matrix $D^{(j)}_{n,e}(s)$. In fact, differentiation of the left hand side of \eqref{eq:expgj}, with respect to $z \in \{a_m , b_m \}$, yields
\be 
\sum_{k=1}^{2\eta^{(j+1)}} \sum_{o = 1}^{j+2} \frac{\partial}{\partial z} d_{n+e_{k,o}} \cdot \prod_{\ell \neq o}^{j+2} d_{n+e_{k,\ell}} 
= \sum_{|e| \leq \lfloor \frac{j}{2} \rfloor +1} \frac{\partial}{\partial z} d_{n+e} \sum_{k=1}^{2\eta^{(j+1)}} \sum_{o = 1}^{j+2} \delta_e(e_{k,o}) \prod_{\ell \neq o}^{j+2} d_{n+e_{k,\ell}} 
\ee
and an analogous formula holds for the partial derivative of the left hand side of \eqref{eq:exphj}. These expressions
determine the matrix entries of $D^{(j)}_{n,e}(s)$. The bound in \eqref{eq:1bdonD} now follows with 
\be \label{eq:Dej}
D^{(j)}_e = \begin{pmatrix}  \sum_{k=1}^{2\eta^{(j+1)}}\sum_{o=1}^{j+2}\delta_e(e_{k,o})&\sum_{k=1}^{2\eta^{(j+1)}}\sum_{o=1}^{j+2}\delta_e(e_{k,o})\\
                       2\sum_{k=1}^{2\xi^{(j+1)}}\sum_{o=1}^{j+2}\delta_e(e_{k,o})&2\sum_{k=1}^{2\xi^{(j+1)}}\sum_{o=1}^{j+2}\delta_e(e_{k,\xi})
                       \end{pmatrix} \, .
\ee
Summing on $e$ yields,
\be
\left\| D^{(j)} \right\|_{\infty} = \left\| \sum_{|e| \leq \lfloor \frac{j}{2} \rfloor +1} D_e^{(j)} \right\|_{\infty} = 8(j+2) \xi^{(j+1)}
\ee
where we have used $\eta^{(j+1)} \leq 2 \xi^{(j+1)}$. This shows that 
\be
\| D(r) \|_{\infty} \leq 8 \sum_{j=0}^{r} |c_{r-j}| \| L(0) \|_2^j (j+2) 3^j \, ,
\ee
and we are done.
\end{proof}

We end this section with an analogue of Theorem~\ref{thm:toda_lrb} for the hierarchy. 
Let $\alpha_t^{(r)}$ denote the dynamics corresponding to the
Toda hierarchy, i.e., $\alpha_t^{(r)}(A) = A \circ \Phi_t^{(r)}$ for all $A \in \mathcal{A}$.

\begin{thm} \label{thm:hier_lrb} Let $r \in \N_0$, $\x \in M$, and $\mu >0$. There exist numbers $C_r$ and $v_r$ for which given 
any observables $A, B \in \A^{(1)}$, the estimate 
\be \label{ineq:hier_gentbd}
|\{\alpha_t^{(r)}(A),B\}({\rm x})| \leq C \sum_{n, m \in \mathbb{Z}} \left( \left\| \fr{\D A }{ \D a_m}  \right\|_\x + \left\| \fr{\D A}{ \D b_m}  \right\|_\x \right) \left( \left\| \fr{\D B}{ \D a_n}  \right\|_\x + \left\| \fr{\D B}{ \D b_n}  \right\|_\x \right) e^{- \mu \left( \lceil |n-m|/ (\lfloor\frac{r}{2}\rfloor+1)\rceil - v_r |t| \right)} 
\ee
holds for all $t\in\R$. Here  $4C_r = \| a\|_{\infty} (1+e^{\mu/ (\lfloor\frac{r}{2}\rfloor+1)}) $ and $v_r$ is as in Theorem~\ref{thm:hth_solest}.
\end{thm}
\begin{proof}
The estimate \eqref{ineq:hier_gentbd} follows as in the proof of Theorem~\ref{thm:toda_lrb} using
the results of Theorem~\ref{thm:hth_solest} as input.
\end{proof}


%
%
\section{Results for the Perturbed Hierarchy} \label{sec:perthier}

The purpose of this section is to demonstrate that the methods from Section~\ref{sec:pert_toda} also apply to the Toda
hierarchy. Many of the proofs follow closely the previous arguments, and so we only sketch the
details. 

To introduce the relevant class of perturbations, we first recall the Hamiltonian formulation of the hierarchy, see e.g.\ Section 1.7 of \cite{GHMT}.
For each $r \in \mathbb{N}_0$, consider the following formal Hamiltonian 
\be \label{eq:hierham}
H_r (\x) = \frac{4}{r+2} \sum_{k \in \mathbb{Z}} \sum_{j=0}^r c_{r-j} \big( \tilde{g}_k^{(j+2)}(\x) -\lambda_{r+2}\big), \qquad
\lambda_r=  \begin{cases}
\frac{1}{2^r}{r\choose r/2}, & r \text{ even},\\
0, & r \text{ odd},
\end{cases}
\ee
with $\tilde{g}_k^{(j+2)}$ as defined in \eqref{eq:gj}.

It can be shown that $H_r$ generates the Toda hierarchy in the sense that for any $A \in \mathcal{A}_0$,
\be
\frac{d}{dt} \alpha_t^{(r)}(A) = \alpha_t^{(r)} \left( \left\{ A, H_r \right\} \right) = \left\{ \alpha_t^{(r)}(A), H_r \right\}
\ee
where $\alpha_t^{(r)}$ is the dynamics associated with the Toda hierarchy as introduced before Theorem~\ref{thm:hier_lrb}. In fact, it is 
shown in e.g.\ \cite[Thm.~1.71]{GHMT} that
\be
\frac{\partial H_r}{ \partial a_n} (\x) = 4 g_n( \x, r) \quad \mbox{and} \quad A_n \frac{\partial H_r}{ \partial b_n}(\x) = 4 h_n(\x, r) 
\ee
and therefore the evolution equations 
\be
\frac{d}{dt} \alpha_t^{(r)}(A_n) = \alpha_t^{(r)} \left( \left\{ A_n , H_r \right\} \right) = \frac{1}{4} \alpha_t^{(r)} \left( A_n \left( \frac{\partial H_r}{\partial b_{n+1}} - \frac{\partial H_r}{\partial b_n} \right) \right)
\ee
and similarly,
\be
\frac{d}{dt} \alpha_t^{(r)}(B_n) = \alpha_t^{(r)} \left( \left\{ B_n , H_r \right\} \right) = \frac{1}{4} \alpha_t^{(r)} \left( A_n \frac{\partial H_r}{\partial a_n} -  A_{n-1} \frac{\partial H_r}{a_{n-1}} \right) 
\ee
follow; compare with \eqref{eq:HTH_1} and \eqref{eq:HTH_2}. 

Note that for $r=0$
\be
H_0(\x)  = \sum_{k\in\Z} \left( 2b_k^2 + 4a_k^2 - 1 \right) .
\ee
is different from \eqref{eq:todahamab}. However, the equations of motion \eqref{eq:toda} are the same as above.

We can now introduce perturbations as in Section~\ref{sec:pert_toda}.
Fix $r \in \N_0$ and let $W: \mathbb{R} \to [0, \infty)$ satisfy $W \in C^2(\R)$. 
Consider the formal Hamiltonian
\be 
H_r^{\w} = H_r + \sum_{n \in \mathbb{Z}} W_n
\ee
where $H_r$ is as in \eqref{eq:hierham} above and $W_n$ is the observable with $W_n(\x) = W( \ln(4a_n^2))$.
 
 The equations of motion corresponding to $H_r^\w$ are 
 \be \label{eq:perthiersola}
 \dot{a}^{\w}_n(t,r) = a_n^{\w}(t,r) \left( g_{n+1}(t,r) -g_n(t,r) \right)
 \ee
 and
 \be \label{eq:perthiersolb}
  \dot{b}^{\w}_n(t,r) = h_n(t,r) - h_{n-1}(t,r) +R_n(t) 
 \ee
 where
 \be
 R_n(t) = \frac{1}{2} \left[ W'(\ln(4a_n^{\w}(t,r)^2)) - W'(\ln(4a_{n-1}^{\w}(t,r)^2))  \right] \, .
 \ee
 
Again, local existence and uniqueness of solutions of \eqref{eq:perthiersola} and \eqref{eq:perthiersolb}, corresponding to initial conditions $\x \in M_0$,
follows from standard results, \cite[Thm.~4.1.5]{Abraham}.
As before, let us denote by $\Phi^{\w}_{t,r}$ the perturbed flow of the Toda hierarchy, i.e., the function 
$\Phi^{\w}_{t,r}(\x) = \left\{ \left(a_n^{\w}(t,r), b_n^{\w}(t,r) \right) \right\}$.
Our arguments apply to the set of initial conditions $M_{b,r} =M_{b,r}(W)$ with bounded trajectories, i.e., for which there exists numbers $C_1, C_2 < \infty$ with
\be
\sup_{t \in \mathbb{R}} \| \Phi^{\w}_{t,r}( \x) \|_M \leq C_1 \quad \mbox{and} \quad \sup_{t \in \mathbb{R}} \sup_{n \in \mathbb{Z}} \frac{1}{|a_n^{\w}(t,r)|} \leq C_2 \, .
\ee

For initial conditions $\x \in M_{b,r}$, there are results similar to the two main estimates from Section~\ref{sec:pert_toda}. 
We first state an analogue of Theorem~\ref{thm:pert_toda_solest}.

\begin{thm} \label{thm:pert_hth_solest} Fix $r \in \N_0$, $W \in C^2(\mathbb{R})$ with $W'' \in L^{\infty}(\mathbb{R})$, and $\x \in M_{b,r}$.
For each $\mu >0$, there exist a number $v^{\w}_r = v^{\w}_r(\mu, \x)$
for which given any $n,m \in \Z$, the estimate 
\be  \label{eq:simp_pert_hier}
\max \left[ \abs{ \fr{\D}{ \D z} a_n^\w(t,r)}, \abs{ \fr{\D}{ \D z} b_n^\w(t,r)} \right] \leq 
e^{- \mu \left( \lceil |n-m|/ (\lfloor\frac{r}{2}\rfloor+1)\rceil - v_r^{\w} |t| \right)} \, ,
\ee
holds for all $t\in\R$. Here $z\in\{a_m,b_m\}$.
\end{thm}
\begin{proof}
We will follow closely the proof of Theorem~\ref{thm:pert_toda_solest} using Theorem~\ref{thm:hth_solest} as input. 
Take $t \geq 0$ and, as before, introduce
\be
F_n^\w(t,r)=\begin{pmatrix} a_n^\w(t,r) \\ b_n^\w(t,r) \end{pmatrix}.
\ee
We will suppress the dependence on $r$. Since $W$ is sufficiently smooth,
$F_n^{\w}$ is differentiable with respect to $z \in \{a_m,b_m\}$ and the bound
\be
\left|\fr{\D}{ \D z} F_n^{\w}(t)\right| \leq \left|\fr{\D}{ \D z} F_n^{\w}(0)\right| + \sum_{j=0}^{r+1} \alpha_j \sum_{|e|\leq \beta_j}\int_0^t D_{e}^{(j)} \left|\fr{\D}{ \D z} F_{n+e}^{\w}(s)\right| ds \, .
\ee
follows as in Theorem~\ref{thm:hth_solest}. Here we have denoted by 
\be
\alpha_j = |c_{r-j}| C_1^{j+1} \quad \mbox{for } 0 \leq j \leq r \quad \mbox{with } \alpha_{r+1} = C_2 \| W'' \|_{\infty}
\ee
and
\be
\beta_j = \lfloor j/2 \rfloor+1\quad \mbox{for } 0 \leq j \leq r \quad \mbox{with } \beta_{r+1} = 1 \, .
\ee
Moreover, for $0 \leq j \leq r$, we have taken $D_e^{(j)}$ as in the proof of Theorem~\ref{thm:hth_solest} and set
\be
D_e^{(r+1)} = \begin{pmatrix} 0 & 0 \\ 1 & 0 \end{pmatrix} \quad \mbox{for } e = -1, \, 0 \quad \mbox{and} \quad D_1^{(r+1)} = \begin{pmatrix} 0 & 0 \\ 0 & 0 \end{pmatrix} \, .
\ee 
Taking $z=a_m$ and iterating yields
\be
\left|\fr{\D}{ \D z} F_n^{\w}(t)\right| \leq \sum_{k=\lceil |n-m|/ (\lfloor\frac{r}{2}\rfloor+1)\rceil}^\infty \fr{t^k}{k!} (D_r^{\w})^k \begin{pmatrix} 1 \\ 0 \end{pmatrix} 
\ee
with 
\be \label{eq:defD2}
D_r^{\w} = \sum_{j=0}^{r+1} \alpha_j D^{(j)}   \quad \mbox{and} \quad  D^{(j)} = \sum_{|e|\leq \beta_j} D_{e}^{(j)} \, .
\ee
A similar estimate holds for the case $z=b_m$.

As in the proof of Theorem~\ref{thm:hth_solest}, the bound \eqref{eq:simp_pert_hier} now follows with
\be
v_r^{\w} = \| D_r^{\w} \|_{\infty} \left( e^{\mu+1} + \frac{1}{\mu} \right) \, .
\ee

Note also that
\be
\| D_r^{\w} \|_{\infty} \leq C_1 \sum_{j=0}^r |c_{r-j}| C_1^j \| D^{(j)} \|_{\infty} + 2 C_2 \|W'' \|_{\infty} \, ,
\ee
and so the estimate from Lemma~\ref{lem:hiervel} applies here as well.
\end{proof}

The interpolation argument proven in Section~\ref{sec:pert_toda} generalizes to the hierarchy as well.  To see this, we first
prove an analogue of Lemma~\ref{lem2}.
\begin{lem}\label{lem:hier2der}
Fix $r \in \N_0$, $\mu >0$, and let $\x \in M_0$. There exists a number $C = C(r, \mu, \x)>0$  and a function $h$, depending on $r$, $\mu$, and $\x$, for which
given any $n,k, \ell \in\Z$, the estimate
\be \label{eq:2derhier}
\max \left[ \left| \fr{\D^2}{ \D z \D \tilde{b}_k} a_n(t,r) \right| , \left| \fr{\D^2}{ \D z \D \tilde{b}_k} b_n(t, r) \right| \right] \leq C e^{-\mu\left(\lceil |n-\ell|/ (\lfloor\frac{r}{2}\rfloor+1)\rceil \right)} e^{-\mu\left(\lceil |n-k|/ (\lfloor\frac{r}{2}\rfloor+1)\rceil \right)} e^{2 \mu v_r |t|} h(t) 
\ee
holds for all $t\in\R$. Here $a_n(t,r)$ and $b_n(t,r)$ are the solutions of \eqref{eq:toda} with initial condition $\x \in M$,
$\fr{\D}{\D \tilde{b}_k} = \fr{\D}{\D b_{k+1}} -\fr{\D}{\D b_k}$, $z \in \{a_\ell , b_\ell \}$, and the number $v_r$ is as in Theorem~\ref{thm:toda_solest}.
The function $h$ grows at most exponentially.
\end{lem}

\begin{proof} 
Again, since this proof follows closely the arguments of Lemma~\ref{lem2}, we will only sketch the details. 

Keeping with the previous notation, it is clear that
\be
\fr{\D^2}{ \D z \D \tilde{b}_k} F_n(t) = 
\sum_{j=0}^r c_{r-j} \sum_{|e|\leq \lfloor j/2 \rfloor+1} \int_0^t \fr{\D}{ \D z} \left(  D_{n,e}^{(j)}(s) \fr{\D}{ \D \tilde{b}_k} F_{n+e}(s) \right) \, ds \, .
\ee
Differentiation with respect to $z$ produces two terms. We begin by estimating the first term, i.e., the one that contains
$\frac{\partial}{\partial z}D_{n,e}^{(j)}(s)$. Using Theorem~\ref{thm:hth_solest}, it is clear that if $|e| \leq \lfloor j/2 \rfloor+1$, then
\be \label{eq:1derhier}
\left| \frac{\partial}{\partial \tilde{b}_k} F_{n+e}(s) \right| \leq e^{\mu} e^{-\mu\left(\lceil |n-k|/ (\lfloor\frac{r}{2}\rfloor+1)\rceil-v_r|s|\right)} \begin{pmatrix} 1 \\ 1 \end{pmatrix} \, .
\ee
For convenience, let us denote by
\be
E(x) = e^{-\mu\left(\lceil |x|/ (\lfloor\frac{r}{2}\rfloor+1)\rceil \right)} \, .
\ee
Using \eqref{eq:1derhier}, one can show that
\be
\sum_{|e|\leq \lfloor j/2 \rfloor+1}\int_0^{|t|}  \left|  \fr{\D}{ \D z} D_{n,e}^{(j)}(s) \fr{\D}{ \D \tilde{b}_k} F_{n+e}(s) \right| \, ds \leq \gamma_j(\mu) E(n-k) E(n- \ell) \int_0^{|t|} e^{2 \mu v_r s} ds \, \begin{pmatrix} 1 \\ 1 \end{pmatrix} 
\ee
where
\be
\gamma_j(\mu) = 8 e^{2 \mu} \| L(0) \|_2^j (j+1) (j+2)  \, .
\ee
This proves that
\bea 
\left| \fr{\D^2}{ \D z \D \tilde{b}_k} F_n(t) \right|  & \leq &  C_{\mu}(r) E(n-k) E(n-\ell) \int_0^{|t|} e^{2 \mu v_r s} \, ds \, \begin{pmatrix} 1 \\ 1 \end{pmatrix}  \nonumber \\
& \mbox{ } & \quad  +  \| L(0) \|_2 \sum_{j=0}^r |c_{r-j}| \| L(0) \|_2^j  \sum_{|e|\leq \lfloor j/2 \rfloor+1} \int_0^{|t|}  D_e^{(j)} \left| \fr{\D^2}{\D z \D \tilde{b}_k} F_{n+e}(s) \right|  \, ds  
\eea
where
\be
C_{\mu}(r) = \sum_{j=0}^r |c_{r-j}| \gamma_j(\mu) \, ,
\ee
and $D_e^{(j)}$ is as in \eqref{eq:1bdonD}. Iteration, as in the proof of Lemma~\ref{lem2}, yields 
\be
\left| \fr{\D^2}{ \D z \D \tilde{b}_k} F_n(t) \right|  \leq  \frac{C_{\mu}(r)}{2 \mu v_r} E(n-k) E(n-\ell) \sum_{m=0}^{\infty} \left( \frac{\| L(0) \|_2}{2 \mu v_r} e^{2 \mu} \right)^m \sum_{p=m+1}^{\infty} \frac{(2 \mu v_r |t|)^p}{p!} \tilde{D}(r)^m 
\begin{pmatrix} 1 \\ 1 \end{pmatrix} \, ,
\ee
where
\be
\tilde{D}(r) = \sum_{j=0}^r |c_{r-j}| \| L(0) \|_2^j  \sum_{|e|\leq \lfloor j/2 \rfloor+1} D_e^{(j)} \, .
\ee
Taking $\infty$-norms, and using \eqref{eq:bigsum}, we see that
\be
\left\| \fr{\D^2}{ \D z \D \tilde{b}_k} F_n(t) \right\|_{\infty}  \leq C E(n-k) E(n-\ell) e^{2 \mu v_r |t|} h(t) 
\ee
where
\be
C = \frac{C_{\mu}(r)}{2 \mu v_r} \quad \mbox{and with } \beta = \frac{e^{2 \mu} \| L(0) \|_2 \| \tilde{D}(r) \|_{\infty}}{2 \mu v_r}  - 1  
\ee
$h$ is as in \eqref{eq:defh}. This proves \eqref{eq:2derhier} and we are done.
\end{proof}

Given Lemma~\ref{lem:hier2der} above, the analogue of Theorem~\ref{thm:pert_toda_interpol} follows
with only minor modifications. We state it below.
\begin{thm}
Fix $r \in \N_0$, $W \in C^2(\mathbb{R})$ with $W', W'' \in L^{\infty}(\mathbb{R})$, and let $\x \in M_{b,r}$. For each $\mu >0$ and any $\epsilon >0$, there are positive numbers 
$C = C(r, \epsilon)$, $D = D(r, \epsilon, \mu, \x, W)$, and $\delta =\delta(r, \epsilon, \mu, \x, W)$ such that
\be \label{eq:interbdhier}
\max \left[ \abs{ \fr{\D}{ \D z} a_n^\w(t,r)}, \abs{ \fr{\D}{ \D z} b_n^\w(t,r)} \right] \leq C G_{\mu_r}( |n-m| ) e^{(\mu+\epsilon) v_r |t|} \left[ 1 + D \left( e^{\delta |t|} -1 \right) \right] 
\ee
holds for all $t\in\R$. Here $v_r = v_r(\mu + \epsilon, \x)$ is as in Theorem~\ref{thm:hth_solest}, $\mu_r = \mu/(\lfloor\frac{r}{2}\rfloor+1)$, and $z\in\{a_m,b_m\}$.
\end{thm}
\begin{proof}
Interpolating as before, it is clear that for any $\x \in M$, the bound
\be \label{eq:1intbdhier}
\begin{split}
\left\| \left[ \fr{\D}{\D z} \tilde{F}^{\w}_n(t) \right] (\x) \right\|_{\infty}  & \leq  \left\| \left[ \fr{\D}{\D z} \tilde{F}_n(t)  \right] (\x) \right\|_{\infty}  +  \sum_{k \in \mathbb{Z}} 
\int_0^{|t|} \left\| \left[  D_k(s;n) \fr{\D}{\D z} \tilde{F}^{\w}_k(s) \right] (\x) \right\|_{\infty} \, ds \\
& \quad \quad + \fr{\| W'\|_{\infty} }{2} \sum_{k, \ell \in \mathbb{Z}} 
\int_0^{|t|}  \left\| \left[ D_{k, \ell}(s;n) \fr{\D}{\D z} \tilde{F}^{\w}_{\ell}(s) \right] (\x) \right\|_{\infty} \, ds
\end{split}
\ee
follows as in the proof of \eqref{eq:1intbd}; here all quantities depend now also on $r$, but we have suppressed this in our notation.

With Theorem~\ref{thm:hth_solest}, it is clear that
\bea
\left\| \left[ \fr{\D}{\D z} \tilde{F}_n(t)  \right] (\x) \right\|_{\infty} & = & \left\| \fr{\D}{\D z} F_n(t)  \right\|_{\infty} \\
& \leq & e^{- (\mu + \epsilon) \left( \lceil |n-m|/ (\lfloor\frac{r}{2}\rfloor+1)\rceil - v_r |t| \right)}  \nonumber \\
& \leq & C_{\epsilon, r} G_{\mu_r}(|n-m|) e^{(\mu + \epsilon) v_r|t|} \, , \nonumber
\eea
where we have set
\be
C_{\epsilon , r} = \sup_{x\geq0} (1+x)^2e^{-\epsilon \left( \lceil x/ (\lfloor\frac{r}{2}\rfloor+1)\rceil \right) } \, , \quad \mu_r = \frac{\mu}{\lfloor\frac{r}{2}\rfloor+1} \, , \quad \mbox{and} \quad v_r = v_r(\x, \mu+\epsilon) \, .
\ee

Similarly, for each $\x \in M_{b,r}$, the matrix appearing in the second term satisfies
\bea
\left\| \left[  D_k(s;n) \right] (\x) \right\|_{\infty} & \leq & C_2 \| W'' \|_{\infty} e^{\mu+ \epsilon}  e^{- (\mu + \epsilon) \left( \lceil |n-k|/ (\lfloor\frac{r}{2}\rfloor+1)\rceil - v_r(s) (|t| -s) \right)}  \\
& \leq & C_2 \| W'' \|_{\infty} e^{\mu+ \epsilon} C_{\epsilon, r} G_{\mu_r}(|n-k|) e^{(\mu+\epsilon) v_r^*(|t|-s) }  \nonumber 
\eea
with
\be
v_r(s) = v_r \left( \Phi_{s,r}^{\w}(\x), \mu+ \epsilon \right) \quad \mbox{and} \quad v_r^* = \sup_{s \in \mathbb{R}} v_r(s) \, .
\ee

Lastly, the bound
\be
\left\| \left[  D_{k, \ell}(s;n) \right] (\x) \right\|_{\infty} \leq \tilde{C}  G_{\mu_r}(|n-\ell|)G_{\mu_r}(|n-k|) e^{2(\gamma+1)(\mu+ \epsilon)v_r^*(|t|-s)}
\ee
follows as in the proof of \eqref{eq:thirdmatrix}; with (possibly) different values of $\tilde{C}$ and $\gamma$. 

Iteration yields \eqref{eq:interbdhier} as in the proof of Theorem~\ref{thm:pert_toda_interpol}.
\end{proof}

%
%
%
%

\section{Locality bounds for more general initial conditions} \label{sec:timedep}
In this section, we return to the class of perturbations considered in Section~\ref{sec:pert_toda},
see also Section~\ref{sec:perthier}.
The goal here is to prove a locality estimate for more general initial conditions.
Let us recall the basic set-up.

Fix $W : \R \to [0, \infty)$ satisfying $W \in C^2( \mathbb{R})$ with $W',W''\in L^\infty(\R)$.
The formal Hamiltonian is given by
\be 
H^\w = H + \sum_{n \in \mathbb{Z}} W_n,
\ee
where $H$ is the Toda Hamiltonian as in \eqref{eq:todahamab}, and for each $n$, the perturbation $W_n$ is taken to be the 
observable $W_n(\x)=W( \ln(4a_n^2))$. The corresponding equations of motion are
\begin{eqnarray}\label{eq:pertab}
\dot{a}_n^{\w}(t) &= &a_n^{\w}(t) \left( b_{n+1}^{\w}(t) - b_n^{\w}(t) \right)\\
\dot{b}_n^{\w}(t) &= & 2 \left( a_n^{\w}(t)^2 -a_{n-1}^{\w}(t)^2 \right)+R_n(t)\nonumber
\end{eqnarray}
where 
\be
R_n(t) = \fr{1}{2}\left[W'( \ln(4a_n^{\w}(t)^2)) - W'( \ln(4a_{n-1}^{\w}(t)^2))\right].
\ee
As we discussed before, local existence and uniqueness of solutions of \eqref{eq:pertab}
corresponding to initial conditions $\x \in M_0$
follows again from Theorem~4.1.5 in \cite{Abraham}.  For $\x \in M_0$, let us again denote by
$\Phi_t^{\w}(\x) = \{ (a_n^{\w}(t), b_n^{\w}(t) ) \}$ the perturbed Toda flow. If we set,
$L(t) = L( \Phi_t^{\w}(\x))$ and $P(t)= P( \Phi_t^{\w}(\x))$ with $L$ and $P$ as in 
 \eqref{eq:defL} and \eqref{eq:defP}, it is easy to check that
 \be
\fr{d}{dt}L(t)=[P(t),L(t)]+R(t) \, ,
\ee
where the multiplication operator $R(t) : \ell^2(\mathbb{Z}) \to \ell^2(\mathbb{Z})$ is defined by
\be
[R(t)f]_n = R_n(t) f_n \, .
\ee
{F}rom this definition, it is clear that 
\be
\| R(t) \|_2 \leq \| W' \|_\infty,
\ee
and hence the operator norm of $R(t)$ is bounded uniformly in both $t$ and the initial condition ${\rm x} \in M_0$. 
 
One can check that $P(t)$ still corresponds to a family of unitary propagators, which we denote by $U(t,s)$.
A short calculation shows then that $\tilde{L}(t)=U(t,s)^*L(t)U(t,s)$ satisfies 
\be
\fr{d}{dt}\tilde{L}(t)=U(t,s)^*R(t)U(t,s) \, .
\ee
In this case, 
\be \label{eq:Lbd}
\| L(t) \|_2 = \| \tilde{L}(t) \|_2 \leq \| \tilde{L}(0) \|_2 + \int_0^{|t|} \| R(s) \|_2 \, ds \leq \| L(0) \|_2 + \| W' \|_\infty |t| \, .
\ee
Since the bound
\be \label{eq:solbd}
\max(\sup_n|a_n^\w(t)|,\sup_n|b_n^\w(t)|)\leq\left\|L(t)\right\|_2 \, 
\ee
holds, \eqref{eq:Lbd} produces a linear bound on the growth of solutions. {F}rom this, the existence of global solutions follows from
Proposition~4.1.22 in \cite{Abraham}.

A result analogous to Theorem~\ref{thm:toda_solest} follows.

\begin{thm} \label{thm:genlrb} Let ${\rm x} = \{ (a_n, b_n) \}_{n \in \mathbb{Z}} \in M_0$ with $a_*=\inf_n |a_n|>0$. Then for any $\mu >0$ and $n,m \in \Z$ the estimate
\be \label{eq:basest2}
\max \left[  \abs{\fr{\D}{ \D z} \ln(a_n^\w(t)^2)}, \abs{ \fr{\D}{ \D z} b_n^\w(t)} \right] \leq \max\left(1,\fr{2}{a_*}\right) e^{- \mu \left( |n-m| - v^\w(t) \right)}
\ee
holds for all $t\in\R$. Here $z\in\{a_m,b_m\}$ and $v^\w(t)$ is an explicit cubic polynomial with $v^\w(0)=0$.
\end{thm}

\begin{proof}
Fix $ {\rm x} \in M_0$ with $a_*>0$. Global existence on $M$ guarantees that for each $n \in \Z$, the function $F_n^\w : \mathbb{R} \to \mathbb{R}^2$ given by
\be
F_n^\w(t)=\begin{pmatrix} \ln(a_n^\w(t)^2) \\ b_n^\w(t) \end{pmatrix}
\ee
is well-defined for all $t\in\R$. In fact, since $a_*>0$, the equations of motion ensure that
\be
a_n^{\w}(t) = a_n^{\w}(0) \exp\left(\int_0^t (b_{n+1}^{\w}(s) - b_n^{\w}(s)) ds \right)
\ee 
preventing a singularity in the logarithm.
It is then clear that \eqref{eq:pertab} implies
\be \label{eq:flowatn2}
F_n^\w(t)= F_n^\w(0)+2 \int_0^t \begin{pmatrix}  b_{n+1}^\w(s) - b_n^\w(s)  \\   a_n^\w(s)^2 -a_{n-1}^\w(s)^2 + \fr{1}{2} R_n(s) \end{pmatrix} \, ds \, .
\ee
For any $z \in \{ a_m, b_m \}$, a relation similar to \eqref{eq:dfnz} holds, i.e.
\be \label{eq:dfnz2}
\fr{\D}{ \D z} F_n^\w(t)  =  \fr{\D}{ \D z} F_n^\w(0) +2 \sum_{|e|\leq 1}\int_0^t D_{n,e}^\w(s) \fr{\D}{ \D z} F_{n+e}^\w(s)ds,
\ee
where
\be
D_{n,e}^\w(s) = \begin{pmatrix} 0 & \delta_{1}(e) -\delta_0(e) \\  \left(a_n^\w(s)^2+\fr{\left\|W''\right\|_\infty}{4}\right) \delta_0(e) - \left(a_{n-1}^\w(s)^2+\fr{\left\|W''\right\|_\infty}{4}\right) \delta_{-1}(e) & 0  \end{pmatrix} 
\ee
The bounds in \eqref{eq:Lbd} and \eqref{eq:solbd} show that
\be
a_n^{\w}(s)^2 \leq \left( \| L(0) \|_2 + \| W' \|_{\infty} |s| \right)^2 \, ,
\ee
and so we find that 
\be \label{eq:gendzfn}
\left| \fr{\D}{\D z}F_n^\w(t) \right|  \leq  \left| \fr{\D}{\D z}F_n^\w(0) \right| + \sum_{|e|\leq 1} \int_0^{|t|} g(s)D_e^\w \left| \fr{\D}{\D z} F_{n+e}^\w(s) \right|ds
\ee
holds for any $t \in \mathbb{R}$. Here we have denoted by
\be
g(s) = 1+\left( \| L(0) \|_2 + \| W' \|_{\infty} |s| \right)^2+\fr{\left\|W''\right\|_\infty}{4}  \quad \mbox{and} \quad D_e^{\w} = \begin{pmatrix} 0 & \delta_{1}(e) +\delta_0(e) \\   \delta_0(e) +  \delta_{-1}(e) & 0  \end{pmatrix}
\ee
Iteration (with $z=a_m$) yields 
\bea
\left| \fr{\D}{\D z}F_n^\w(t) \right|  &\leq&  \sum_{k=|n-m|}^\infty \sum_{|e_1|\leq 1}\cdots\sum_{|e_k|\leq 1}\int_0^{|t|}\int_0^{t_1}\cdots\int_0^{t_{k-1}}g(t_1)\cdots g(t_k)\\
& &\times\ D_{e_1}^\w\cdots D_{e_k}^\w\left| \fr{\D}{\D z}F_{n+e_1+\cdots+e_k}^\w(0) \right|dt_k\cdots dt_1\\
&\leq& \fr{2}{|a_m|}\sum_{k=|n-m|}^\infty\fr{h(t)^k}{k!}(D^\w)^k\begin{pmatrix} 1 \\ 0 \end{pmatrix},
\eea
where $h(t)=2\int_0^{|t|}g(s)ds$ and
\be
D^\w=\fr{1}{2}\sum_{|e|\leq 1}D_e^\w=\begin{pmatrix}0 & 1 \\ 1 & 0 \end{pmatrix} \, .
\ee
In the case that $z=b_m$, it is clear that
\be
\left| \fr{\D}{\D z}F_n^\w(t) \right| \leq \sum_{k=|n-m|}^\infty\fr{h(t)^k}{k!}(D^\w)^k\begin{pmatrix} 0 \\ 1 \end{pmatrix} \, .
\ee
Taking the infinity norm, we get for any $\mu>0$
\bea 
\left\| \fr{\D}{\D z}F^{\w}(t) \right\|_\infty  &\leq& \max\left(1,\fr{2}{|a_m|}\right)\sum_{k=|n-m|}^\infty\fr{h(t)^k}{k!}\\
&\leq &\max\left(1,\fr{2}{a_*}\right) e^{- \mu \left( |n-m| - v^\w(t) \right)}, \nonumber
\eea
where $v^\w(t)=h(t)\left( e^{\mu+1}+\mu^{-1}\right)$.
\end{proof}

Again, everything extends to the hierarchy. Since the statements are clear, we do not rewrite them for the sake of brevity.

{\bf Acknowledgements:} The authors would like to thank Bruno Nachtergaele for suggesting the problem
discussed in this article and for several productive discussions. Also, the authors would like to 
acknowledge the hospitality of the Erwin Schr\"odinger Institute in Vienna, Austria
where parts of this paper were discussed.

\appendix

\section{Existence of Bounded Solutions for Hamiltonian Systems} \label{sec:app}

In this appendix we want to look at a general Hamiltonian system with nearest neighbor interaction:
\be
H({\rm x})=\sum_{n\in \Z}\left( \frac{p_n^2}{2} + V(q_{n+1} -q_{n}) \right),
\ee
where ${\rm x}=\{(p_n,q_n)\}_{n\in\Z}$. Let $V\in C^2(\R)$ with $V(x)\ge 0$ and $V(0)=V'(0)=0$ such that $0$
is a fixed point of the system. Since we want to obtain bounded solutions we will assume that
our interaction potential is confining in the sense that $V(x)\to+\infty$ as $x\to\pm\infty$. Since the
uniform motion $q_n(t) = q^0 + p^0 t$ (with $q^0,p^0$ some real constants) of the system is unbounded
we switch to relative coordinates $r_n=q_{n+1}-q_n$ in which the equation of motions read
\be\label{ghs}
\dot{r}_n = p_{n+1} - p_n,\qquad
\dot{p}_n = V'(r_n) - V'(r_{n-1}).
\ee
We will consider this system in the Hilbert space $\mathcal{X}=\ell^2(\Z) \times \ell^2(\Z)$. 

\begin{thm}\label{Lemma:Gerald}
Suppose $V\in C^2(\R)$ such that $0$ is a unique global minimum with $V(0)=V'(0)=0$,
$V''(0)>0$ and $V(x)\to+\infty$ as $|x|\to\infty$.

Then the system \eqref{ghs} has a unique global solution in
$\mathcal{X}$ for which the energy
\be
H(p,r)= \sum_{n\in \Z}\left( \frac{p_n^2}{2} + V(r_n) \right)
\ee
is finite and conserved. This solution is $C^1$ with respect to the initial condition.
Moreover, $0$ is a stable fixed point and all solutions satisfy
$\|(p(t),r(t))\|_2\le C$ as well as $\|(p(t),r(t))\|_\infty\le C$, where the constant $C$ depends only
on the initial condition.
\end{thm}

\begin{proof}
First of all note that by our assumption on $V$ we can find constants $c_R$ and
$C_R$ for every $R>0$ such that $|V(x)| \le C_R x^2$, $|V'(x)| \le C_R |x|$ and $|V(x)| \ge c_R x^2$ for $|x|\le R$.

In particular, for $r\in\ell^2(\Z)$ with $\|r\|_2 \le R$ we have $\|V'(r)\|_2 \le C_R \|r\|_2$ and it follows that
the map $r_n\mapsto V'(r_n)$ is $C^1$ on $\ell^2(\Z)$. Since the same is true for the shift operator $x_n \mapsto x_{n-1}$,
our vector field is $C^1$ and local existence and uniqueness follow
from standard results \cite[Thm.~4.1.5]{Abraham}. This also implies that the flow is $C^1$ with respect to the initial
condition \cite[Lem.~4.1.9]{Abraham}. Moreover, $|H(p,r)| \le C_R \|(p,r)\|_2$ implies that $H$
is finite on $\mathcal{X}$ and a straightforward calculation shows that it is conserved by the flow.

Moreover, $H(p,r)\ge c_1\, \epsilon$ for $\|(p,r)\|_2=\epsilon$ and $\epsilon\le 1$. Hence \cite[Thm.~4.3.11]{Abraham}
shows that $0$ is a stable fixed point.

Finally, if $V(x)\to+\infty$ there is a constant $M_E$ such that $|x|\le M_E$ whenever $|V(x)| \le E$. Hence 
setting $E=H(p(0),r(0))$ we have $H(p(t),r(t))=E$ implying $\|p(t)\|_2 \le \sqrt{2E}$ and $\|r(t)\|_\infty\le M_E$.
But this implies $\|r(t)\|_2 \le c_{M_E}^{-1} \|V(r)\|_2 \le c_{M_E}^{-1} \sqrt{E}$.
Hence our vector field remains bounded along integral curves on finite $t$ intervals and hence all
solutions are global in time by  \cite[Prop.~4.1.22]{Abraham}.
\end{proof}

Note that, using $q_n(t) =q_n(0) + \int_0^t p_n(s) ds$, for our original variables we get
\be
\|q(t)\|_2 \le \|q(0)\|_2 + C |t|, \quad \|q(t)\|_\infty \le \|q(0)\|_\infty + C |t|.
\ee
Moreover, clearly the Toda potential $V(r)=e^{-r}+r-1$ satisfies the above assumptions.

\begin{thm}
Let $\x=(p,r) \in \mathcal{X}$ and $\mu >0$. There exists a number $v=v(\mu,\x)$
for which given any $n,m \in \mathbb{Z}$, the bound 
\be
\max \left[ \abs{ \fr{\D}{ \D z} p_n(t)}, \abs{ \fr{\D}{ \D z} r_n(t)} \right] \leq 
C \, e^{- \mu \left( |n-m| - v |t| \right)} \, ,
\ee
holds for all $t\in\R$ and each $z\in\{p_m,r_m\}$, where
\be
C = C(\x) = \max\left( \sup_{(t,n)\in\R\times\Z} |V'(r_n(t))|^{1/2}, 1\right).
\ee
In fact, one may take
\be
v = 2 C \left(e^{\mu+1}+\fr{1}{\mu}\right).
\ee
\end{thm}

\begin{proof}
Fix $\x \in\mathcal{X}$. Our previous theorem guarantees that for each $\x \in M$ and $n \in \mathbb{Z}$, the
function $F_n : \mathbb{R} \to \mathbb{R}^2$ given by
\be
F_n(t; \x)=\begin{pmatrix} r_n(t) \\ p_n(t) \end{pmatrix} \,
\ee
is well-defined and differentiable with respect to each $z \in \{ p_m, r_m \}$.
When convenient, we will suppress the dependence of $F_n$ on $\x$.
Using the equations of motion, i.e.\ \eqref{ghs}, it is clear that
\be
F_n(t)= F_n(0)+\int_0^t \begin{pmatrix}  p_{n+1}(s) - p_n(s) \\  V'(r_n(s)) - V'(r_{n-1}(s)) \end{pmatrix} \, ds \, .
\ee
Differentiating with respect to $z$ we obtain
\be
\fr{\D}{ \D z} F_n(t) = \fr{\D}{ \D z} F_n(0) +\sum_{|e|\leq 1}\int_0^t D_{n,e}(s) \fr{\D}{ \D z} F_{n+e}(s) \, ds,
\ee
with 
\be
D_{n,e}(s) = \begin{pmatrix}0 & \delta_1(e)-\delta_0(e) \\ V'(r_n) \delta_0(e) - V'(r_{n-1}) \delta_{-1}(e) & 0  \end{pmatrix} \, .
\ee
Hence
\be
\left|\fr{\D}{ \D z} F_n(t)\right| \leq \left|\fr{\D}{ \D z} F_n(0)\right| + \sum_{|e|\leq 1}\int_0^{|t|} D_{e} \left|\fr{\D}{ \D z} F_{n+e}(s)\right| ds,
\ee
where 
\be
D_e = \begin{pmatrix} 0 & \delta_1(e)+\delta_0(e) \\ C^2 (\delta_0(e) + \delta_{-1}(e)) & 0  \end{pmatrix} \, .
\ee

Let us now consider the case that $z = r_m$, i.e.,
\be
\frac{ \partial}{\partial r_m}F_n(0) = \begin{pmatrix} 1 \\ 0 \end{pmatrix} \delta_m(n) \, .
\ee
In this case, iteration yields
\bea
\left|\fr{\D}{ \D r_m} F_n(t)\right| &\leq& \sum_{k=0}^\infty  \fr{|t|^k}{k!} \sum_{|e_1|\leq 1} \cdots \sum_{|e_k| \leq 1} \delta_{m+e_1+\cdots + e_k}(n) D_{e_1}\cdots D_{e_k} \begin{pmatrix} 1 \\ 0 \end{pmatrix} \\ \nonumber
&\le&\sum_{k=|n-m|}^\infty \fr{ |2 C t|^k}{k!} D^k \begin{pmatrix} 1 \\ 0 \end{pmatrix} \, ,
\eea
where we have set
\be
D = \frac{1}{2C} \sum_{|e| \leq 1} D_e =  \begin{pmatrix} 0 & C^{-1} \\  C & 0  \end{pmatrix}.
\ee
Again, convergence is guaranteed since $\fr{\D}{ \D z} F_n(t)$ is continuous and thus bounded on compact time intervals.
Using $D^2= \idty$, one obtains
\be
\left\| \fr{\D}{\D r_m}F_n(t) \right\|_{\infty} \leq C \sum_{k=|n-m|}^\infty\fr{|2C t|^k}{k!}
\ee
The rest follows as in Theorem~\ref{thm:toda_solest}.
\end{proof}

Note that in Flaschka variables \eqref{def:ab} the equations of motion read
\be \label{eq:toda_2}
\dot{a}_n(t) = a_n(t) \left( b_{n+1}(t) - b_n(t) \right) \quad \mbox{and} \quad 
\dot{b}_n(t) = -\frac{1}{2}\left(V'(- \ln(4a_n^2)) - V'(- \ln(4a_{n-1}^2)) \right)\, ,
\ee
and $(p,r)$ will be bounded if and only if $(a,a^{-1},b)$ are bounded. The fixed point in these new coordinates
is $(a_0,b_0)=(\frac{1}{2},0)$ and $(p,r)\in \mathcal{X}$ if and only if $(a,b)-(a_0,b_0)=(a-\frac{1}{2},b) \in \mathcal{X}$.

Finally, by
\be
a_n(t) = a_n(0) \exp\left(\int_0^t (b_{n+1}(s) - b_n(s)) ds \right)
\ee
the sign of $a_n$ is preserved under the flow.

\end{document}